%% file: main_arxiv.tex
\pdfoutput=1

\documentclass[sigconf,nonacm]{aamas}

\usepackage{algorithm}
\usepackage{algpseudocode}
\usepackage{makecell} 

\microtypesetup{expansion=false}

\theoremstyle{plain}
\newtheorem{proposition}{Proposition}

\setcopyright{none}
\makeatletter
\AtBeginDocument{%
  \fancypagestyle{firstpagestyle}{%
    \fancyhf{}%
    \renewcommand{\headrulewidth}{\z@}%
    \renewcommand{\footrulewidth}{\z@}%
    \fancyfoot[C]{\footnotesize\thepage}}%
  \fancypagestyle{standardpagestyle}{%
    \fancyhf{}%
    \renewcommand{\headrulewidth}{\z@}%
    \renewcommand{\footrulewidth}{\z@}%
    \fancyhead[L]{\@headfootfont TANDEM: Balancing Centralized and Decentralized Critics}%
    \fancyhead[R]{\@headfootfont \shortauthors}%
    \fancyfoot[C]{\footnotesize\thepage}}%
  \pagestyle{standardpagestyle}}
\makeatother

\title{TANDEM: Balancing Centralized and Decentralized Critics to Address the
Information Gap in Complex Multi-Agent Environments}

\author{Alireza Saleh Abadi}
\affiliation{%
  \department{School of Computing}
  \institution{University of Nebraska-Lincoln}
  \city{Lincoln}
  \state{Nebraska}
  \country{USA}}
\email{asalehabadi2@huskers.unl.edu}

\author{Leen-Kiat Soh}
\affiliation{%
  \department{School of Computing}
  \institution{University of Nebraska-Lincoln}
  \city{Lincoln}
  \state{Nebraska}
  \country{USA}}
\email{lksoh@unl.edu}

\author{Adam Eck}
\affiliation{%
  \department{Department of Computer Science}
  \institution{Oberlin College}
  \city{Oberlin}
  \state{Ohio}
  \country{USA}}
\email{aeck@oberlin.edu}

\author{Prashant Doshi}
\affiliation{%
  \department{Institute for AI \& School of Computing}
  \institution{University of Georgia}
  \city{Athens}
  \state{Georgia}
  \country{USA}}
\email{pdoshi@uga.edu}

\renewcommand{\shortauthors}{Saleh Abadi, Soh, Eck, and Doshi}

\begin{document}

\begin{abstract}
\input{sections/0-abstract}
\end{abstract}


\maketitle

\section{Introduction}\label{sec:introduction}
\input{sections/1-introduction}

\section{Related Work}\label{sec:related}
\input{sections/8-related}

\section{Background}\label{sec:background}
\input{sections/2-background}

\section{Problem Formulation}\label{sec:problem}
\input{sections/3-problem}

\section{Methodology}\label{sec:methodology}
\input{sections/4-methodology}

\section{Experiments}\label{sec:experiments}
\input{sections/5-experiments}

\section{Results}\label{sec:results}
\input{sections/6-results}

\section{Ablations}\label{sec:ablations}
\input{sections/7-ablations}

\section{Conclusion}\label{sec:conclusion}
\input{sections/9-conclusion}

\begin{acks}
\input{sections/10-acknowledgement}
\end{acks}

\bibliographystyle{ACM-Reference-Format}
\bibliography{references}

\clearpage
\appendix
\setcounter{table}{0}
\setcounter{figure}{0}
\renewcommand{\thesection}{S\arabic{section}}
\renewcommand{\thetable}{S\arabic{table}}
\renewcommand{\thefigure}{S\arabic{figure}}
\renewcommand{\theHsection}{S\arabic{section}}
\renewcommand{\theHtable}{S\arabic{table}}
\renewcommand{\theHfigure}{S\arabic{figure}}

\section*{Supplementary Material}
\input{sections/11-supplementary}

\end{document}

%% file: sections/0-abstract.tex
Cooperative multi-agent reinforcement learning under centralized training with
decentralized execution (CTDE) shapes decentralized policies with a centralized
critic that observes the global state, then removes the critic at execution. The
policies are optimized to rely on information they will not have when they act, an
information gap between training and execution. It widens under partial
observability, observation noise, and openness, where the set of agents and tasks
present changes during an episode. Existing methods transfer centralized guidance
in one direction only, removing or distilling the critic once, so the policies
never rehearse acting without the critic while they can still be corrected. We
propose TANDEM, which keeps one centralized critic and reads it two ways at each
training step: a centralized read from the global state, and a decentralized read
restricted to the information each agent will have at execution. A parameter-free
gate compares the two reads for each agent and trains the policy on the
decentralized read where they agree, reinstating the centralized read only where
the local view is misleading. A shared target pulls the decentralized read
toward the centralized one. We formalize the setting as a task- and agent-open
decentralized partially observable Markov decision process (TAO-Dec-POMDP) and evaluate TANDEM on three cooperative domains, each isolating one form
of the information gap: discovery for partial observability, sampling
for observation noise, and openness in wildfire suppression. We compare against
six CTDE baselines and ablate the gate and the shared target; the supplementary
material holds the full variant sweep, diagnostics, and learning curves. We
measure return at execution, once the centralized critic is removed. TANDEM is
the best method on eight of the nine setups and never significantly worse than
any baseline. Where the information gap is widest, it significantly beats every
baseline. Removing the gate reduces execution return most under openness, where the gate
helps most, reinstating the centralized read exactly where the local view
is misleading.

%% file: sections/1-introduction.tex
Among the paradigms for training multiple cooperating agents, centralized
training with decentralized execution (CTDE)
\citep{oliehoek2008optimal,kraemer2016rehearsal} is dominant
\citep{albrecht2024marl,papoudakis2021benchmark}. In cooperative multi-agent
reinforcement learning (MARL), CTDE places a centralized critic during
training that observes the global state and the joint actions of all agents,
and uses it to update the decentralized policies. Centralized information stabilizes learning and, with a joint-action critic, resolves credit
assignment~\citep{lowe2017maddpg,foerster2018coma,mappo,amato2024ctde}.
At execution, the centralized critic is removed, and each agent acts only on
its local observation history, so the system does not rely on
global information, which may be unavailable at real-world deployment. The decentralized policies are therefore
optimized under an information regime they never face at execution, a
train-execution asymmetry that creates an \emph{information
gap}: the discrepancy between the global state $s^t$ the centralized critic
conditions on and each agent's own information (its history $\tau_i^t$) available
at execution. This gap is where execution failures concentrate: under
partial observability, observation noise, or openness~\citep{eck2023openness},
the decentralized policies have no centralized critic to fall back on and the behaviors learned become brittle.

The information gap matters because the deployment objective differs from the
training objective: at deployment each agent must act autonomously, without the
centralized critic, whereas CTDE optimizes performance while the critic is
present. Under constant guidance the policies can over-rely on the centralized
critic, settling into coordination they cannot sustain once it is removed. In other words, training performance
overstates the autonomous performance deployment requires~\citep{oliehoek2016decpomdp}.

Existing CTDE methods address the information gap with transfer mechanisms that
are \emph{monotonic}: centralized information is provided and then reduced or
removed exactly once. Standard CTDE algorithms keep the centralized critic
present at every training step and remove it only at
execution~\citep{mappo,lowe2017maddpg,rashid2018qmix}. Teacher--student methods
transfer a centralized teacher's knowledge to decentralized students in a
single stage~\citep{zhao2024ctds,lin2019cesma,chen2024ptde,li2025magpo};
Kickstarting decays the weight of a teacher distillation loss over
training~\citep{schmitt2018kickstarting}; and Message-Dropout masks parts of
the centralized critic input per sample~\citep{kim2019messagedropout}. These
methods share one weakness: because their transfer is monotonic,
the decentralized policies never train on the execution regime while still
being correctable, in that same training, wherever their local view leads them
astray.

We propose \emph{TANDEM}, a CTDE training scheme that balances a centralized
and a decentralized read of one critic. Rather than training two value networks,
TANDEM keeps one centralized critic and reads it two ways at each step: 
a \emph{centralized read} from the global state, and
a \emph{decentralized read} restricted to only the information each agent will have at
execution~\citep{lyu2021contrasting}. A
parameter-free gate compares the two reads for every agent at every step, mixing
the centralized and decentralized advantage estimates they yield. Where
they agree, the decentralized read already carries the right sign of the
advantage (whether the action did better or worse than expected), so the
centralized critic is withdrawn and the policy trains on the execution regime
(the decentralized read, restricted to the information the actor will have at
execution); where they disagree, the agent's local view is misleading, so the centralized
critic is reinstated to correct it. A shared target pulls the
decentralized read toward the centralized one, so the centralized critic is
increasingly withdrawn as training proceeds, continuously and per agent rather
than on a fixed schedule. As training proceeds, the two reads increasingly agree, so the
decentralized policies act almost as well without the centralized critic as with it; we
track this agreement as a diagnostic, not a convergence point, since under openness and
observation noise, execution return keeps fluctuating across training and no single checkpoint is cleanly ``converged''.

This paper makes five contributions toward handling the information gap so that
decentralized policies perform autonomously. \textbf{First}, we identify monotonic
transfer as the weakness common to existing methods, and we frame our approach (withdrawing and reinstating the centralized critic per agent and
per step according to whether a decentralized read agrees with it) as a
distinct design point. \textbf{Second}, we give a formal model of the
setting as a task- and agent-open Dec-POMDP (TAO-Dec-POMDP;
Section~\ref{sec:problem}), defining the information gap precisely, rather than
by analogy, and characterizing its effect in Proposition~\ref{prop:gap}. \textbf{Third}, we instantiate TANDEM using
standard multi-agent proximal policy optimization (MAPPO)
training~\citep{mappo}, specifying the two critic reads,
the parameter-free gate that mixes them, and the shared target, and establish three
properties that characterize the gate and the information gap
(Section~\ref{sec:method-theory}).
\textbf{Fourth}, we evaluate TANDEM on three cooperative domains, each a distinct form
of the information gap: \emph{discovery} and \emph{sampling} in the
Vectorized Multi-Agent Simulator (VMAS)~\citep{bettini2022vmas}, targeting
partial observability and observation noise, and an openness-aware wildfire
suppression environment~\citep{eck2020wildfire}, targeting openness. \textbf{Finally}, we compare TANDEM
against six CTDE baselines and ablate the gate and the shared
target under each form of the information gap. TANDEM's advantage is largest where the
information gap is widest, and ablations show the gate matters most under openness,
where removing it lowers~execution~return.


%% file: sections/8-related.tex
We review prior work by how it trains decentralized policies under the
information gap. Methods either condition on centralized information that
execution removes, transfer or fade that information through a one-way schedule,
or vary the task rather than the information the policy is optimized under; none
withdraws and reinstates the centralized critic.

\textbf{CTDE Methods.}  These train decentralized policies with centralized information that is removed at
execution~\citep{amato2024ctde}. Centralized-critic actor--critic methods, including
MADDPG~\citep{lowe2017maddpg}, COMA~\citep{foerster2018coma}, MAPPO~\citep{mappo}, and
FACMAC~\citep{peng2021facmac}, condition a critic on global information (the global
state, and in the action-conditioned cases the joint actions) during training; value-factorization methods, including VDN~\citep{sunehag2018vdn},
QMIX~\citep{rashid2018qmix}, QTRAN~\citep{son2019qtran}, and QPLEX~\citep{wang2021qplex},
decompose a joint value learned from global information into per-agent utilities. This
conditioning gives each agent a stable target and assigns credit across agents, which
independent learners cannot do, but it is the source of the information gap
(Section~\ref{sec:introduction}). The gap
is intrinsic rather than a tuning issue: even a fully converged centralized
critic is not strictly beneficial relative to one conditioned on local
histories~\citep{lyu2021contrasting,lyu2023centralized}.

\textbf{Transfer and Cue Removal.}
\label{sec:rw-teacher}  Teacher--student transfer and cue-removal methods bridge centralized training to
decentralized execution by monotonic transfer: centralized information enters as a
teacher, a distillation target, or an extra input, and is faded or pruned on a schedule
until the policy acts without it. Teacher--student methods regress decentralized students
onto centralized values~\citep{zhao2024ctds,lin2019cesma,chen2024ptde,li2025magpo},
building on policy distillation and Kickstarting~\citep{rusu2016distillation,schmitt2018kickstarting};
fading and pruning methods mask or reduce centralized critic inputs and communication over
training~\citep{kim2019messagedropout,zhou2023cadp,li2023explicit,liu2024tacit}; and interactive
imitation corrects the learner on the states it visits with teacher
targets~\citep{ross2011dagger}. In each case the policy is shaped by information that
execution lacks, and withdrawing it on a one-way schedule leaves the policy
formed around having had it.

\textbf{Openness, Ad Hoc Teamwork, and Robustness.}
\label{sec:rw-openness}  They define the conditions we evaluate under rather
than train for the information gap. Open agent systems formalize environments whose agent,
task, and frame sets change over time~\citep{eck2023openness}, with task openness studied
in MARL~\citep{anilmohito} and wildfire suppression as a standard
benchmark domain~\citep{eck2020wildfire,patino2025inaugural}. Openness breaks the fixed-team and
stationarity assumptions that standard CTDE credit assignment relies
on~\citep{abadi2025challenges}. Ad hoc teamwork coordinates with unknown teammates that
enter and leave~\citep{stone2010ad,rahman2021gpl,wang2024n}, and robustness formalizes
observation noise as perturbation of the state observation~\citep{zhang2020robust,he2023robust}.
Curriculum, multi-task, and open-ended learning vary the task rather than the information
regime~\citep{bengio2009curriculum,teh2017distral,wang2019poet}. 

Across these families the decentralized policies are optimized while centralized information
is present, so each depends on information that execution removes. TANDEM instead
trains each policy on its own decentralized read wherever the two reads agree, so the policy
is optimized on the execution regime, and reinstates the centralized read only where the
local view is misleading, using the critic to correct the policy's own actions there rather
than supplying a privileged target to imitate.

%% file: sections/2-background.tex
This section reviews what TANDEM builds on: open agent systems, the open
Dec-POMDP that Section~\ref{sec:problem} extends, the deployment conditions TANDEM
must withstand, and the CTDE paradigm and the MAPPO actor-critic.

\subsection{Open Agent Systems}
\label{sec:bg-openness}

Openness refers to multi-agent environments whose composition changes over
time~\citep{eck2023openness}. Prior work identifies three forms: agent openness
(AO), agents entering and leaving and altering team composition; task openness
(TO), work items appearing and disappearing during operation; and frame openness
(FO), agents' capabilities changing~\citep{eck2023openness}. We address
AO and TO and leave FO to future work.

\subsection{The Open Dec-POMDP}
\label{sec:bg-opendecpomdp}

The open Dec-POMDP~\citep{cohen2017open} extends the Dec-POMDP by letting the set
of agents change within an episode. We first recall the Dec-POMDP, then present
the open Dec-POMDP and the four components in which the two differ.

\textbf{Dec-POMDP.} A Dec-POMDP (decentralized partially observable Markov decision
process)~\citep{oliehoek2016decpomdp} models a cooperative multi-agent task under
partial observability. It has a set of agents, a finite set of states
$\mathcal{S}$, and for each agent $i$ a finite action set $\mathcal{A}_i$ and a
finite observation set $\mathcal{O}_i$. It adds a transition function, an
observation function, a reward $R$ shared by all agents, a horizon $h$, and an
initial state distribution $b_0 \in \Delta(\mathcal{S})$, writing $\Delta(X)$ for
the set of probability distributions over $X$. At each step $t$ every agent
selects an action $a_i^t$, the environment transitions, and the agents receive
$R(s^t, \mathbf{a}^t)$ for the state $s^t$ and the joint action $\mathbf{a}^t$.
Agent $i$ never sees $s^t$; it receives a local observation $o_i^t$ and
conditions its action on its local action-observation history $\tau_i^t$, the
sequence of actions it has taken and observations it has received. The objective
is the expected cumulative reward
$\mathbb{E}_{\pi}[\sum_{t=0}^{h-1} \gamma^t R(s^t, \mathbf{a}^t)]$ under a joint
policy $\pi$, with discount $0 \leq \gamma \leq 1$~\citep{oliehoek2016decpomdp}.

\textbf{Open Dec-POMDP.} The open Dec-POMDP is the tuple:
\begin{equation}
\label{eq:open-decpomdp}
\langle \mathcal{N}, \mathcal{I}, \mathcal{S}, \mathcal{A}, \mathcal{O}, \phi, r, b_0, T \rangle ,
\end{equation}
where $\mathcal{N} = \{1, \ldots, n\}$ is a set of $n$ agents,
$\mathcal{I} \subseteq 2^{\mathcal{N}}$ is a set of coalitions (subsets of
$\mathcal{N}$ that may operate together),
$\mathcal{S}$ is a finite set of world states, and
$\mathcal{A} = \{\mathcal{A}_i \mid i \in \mathcal{N}\}$ and
$\mathcal{O} = \{\mathcal{O}_i \mid i \in \mathcal{N}\}$ are the action and
observation sets. A coalition $I$ acts and observes jointly, through
$\mathcal{A}_I = \times_{i \in I} \mathcal{A}_i$ and
$\mathcal{O}_I = \times_{i \in I} \mathcal{O}_i$. The remaining four components
are the dynamics model $\phi$, given below; the reward model
$r : \mathcal{S} \times \mathcal{A}_I \to \mathbb{R}$, which a Dec-POMDP writes
$R$; the initial belief $b_0 \in \Delta(\mathcal{S} \times \mathcal{I})$ over
states and coalitions; and the planning horizon $T$~\citep{cohen2017open}. Four
components differ from the Dec-POMDP. The set $\mathcal{N}$ becomes a population
from which the operating coalition $I^t \in \mathcal{I}$ is drawn at each step, and
only an active agent $i \in I^t$ acts and observes, so the joint action is drawn
from $\mathcal{A}_{I^t}$ rather than over all agents. The transition and
observation functions collapse into a single dynamics model:
\begin{equation}
\label{eq:open-dynamics}
\phi(s', \mathbf{o}', I' \mid \mathbf{a}, s, I) ,
\end{equation}
which assigns a probability to the next state $s'$, joint observation
$\mathbf{o}'$, and next coalition $I'$ (primes denote next-step quantities, those
at $t{+}1$) given the current state $s$, coalition $I$,
and joint action $\mathbf{a} \in \mathcal{A}_I$, so agents enter and leave under
$\phi$. The objective is the undiscounted sum of rewards over the horizon $T$,
and the initial distribution draws a state and coalition together,
$(s^0, I^0) \sim b_0$.

Eq.~\eqref{eq:open-decpomdp} does \emph{not} represent TO: the coalition
varies, but nothing in the tuple makes tasks appear or disappear, so the model
covers AO alone. Section~\ref{sec:problem} extends it to handle both AO and TO.

\subsection{Partial Observability and Observation Noise}
\label{sec:bg-perception}

Partial observability and observation noise are the deployment conditions that
limit what an agent can use at execution. Under partial observability, as in both
models of Section~\ref{sec:bg-opendecpomdp}, each agent sees only its local
view, never the global state $s^t$. Observation noise further perturbs this view at execution, so the
observation an agent acts on differs from what it has learned.

\subsection{CTDE and MAPPO}
\label{sec:bg-ctde}\label{sec:bg-mappo}

CTDE trains with global information that the environment reveals to the learner but not to any single agent, restricting each agent to local information at execution~\citep{oliehoek2008optimal,kraemer2016rehearsal}. During training, a
centralized critic observes the global state $s^t$ and the joint action
$\mathbf{a}^t$ and uses it to update the decentralized
policies~\citep{lowe2017maddpg,foerster2018coma,amato2024ctde}. A centralized
critic stabilizes learning, and a joint-action critic can also resolve credit
assignment among agents. The state-only critic of Section~\ref{sec:bg-mappo}
gives the first without the second. TANDEM requires a
state-only critic: it evaluates the critic from one agent's masked local view
(Section~\ref{sec:methodology}), which a joint-action critic, needing every
agent's action, cannot provide. At execution, the
centralized critic is removed and each policy acts only on its
local action-observation history $\tau_i^t$; this train-execution asymmetry is
the source of the information gap identified in Section~\ref{sec:introduction}.

MAPPO~\citep{mappo} is a CTDE actor-critic algorithm that applies proximal policy
optimization (PPO)~\citep{ppo} to the cooperative multi-agent setting: each agent
$i$ has a decentralized policy $\pi_\theta(a_i^t \mid \tau_i^t)$ conditioned on
its local history, with parameters $\theta$ shared across agents, and a
centralized critic $V_\psi(s^t)$ with parameters $\psi$, conditioned on the
global state, estimates returns. From $V_\psi$, MAPPO forms an advantage estimate
$\hat{A}^t$ using generalized advantage estimation (GAE)~\citep{gae} and updates
each policy by maximizing the clipped surrogate objective:
\begin{equation}\label{eq:mappo}
\begin{aligned}
L(\theta) &= \mathbb{E}_{i,t}\bigl[\min\bigl(\omega_i^t(\theta)\,\hat{A}^t,\;
\mathrm{clip}(\omega_i^t(\theta), 1{-}\epsilon, 1{+}\epsilon)\,\hat{A}^t\bigr)\bigr], \\[4pt]
\omega_i^t(\theta) &= \frac{\pi_\theta(a_i^t \mid \tau_i^t)}{\pi_{\theta_{\mathrm{old}}}(a_i^t \mid \tau_i^t)},
\end{aligned}
\end{equation}
where $\omega_i^t(\theta)$ is agent $i$'s probability ratio at step $t$ between
the updated and the previous policy, the expectation runs over agents and steps,
and $\epsilon$ is the clipping parameter. MAPPO's critic is state-only, so the
advantage $\hat{A}^t$ carries global information through $V_\psi(s^t)$ rather than
through the joint action. TANDEM uses the update of
Eq.~\eqref{eq:mappo} with an advantage mixed from a centralized
and a decentralized read of the same critic (Section~\ref{sec:methodology}).

%% file: sections/3-problem.tex
We first introduce the \emph{task- and agent-open Dec-POMDP} (TAO-Dec-POMDP), an extension
of the open Dec-POMDP~\citep{cohen2017open} that adds TO to the AO it
already models. Because the set of tasks present changes over
an episode and the actions that address a task exist only while that task is present, the actions
available to an agent depend on the current tasks rather than a fixed set. We obtain the model by
carrying the tasks in the state and letting each agent's available actions depend on it. The TAO-Dec-POMDP is the tuple:
\begin{equation}
\label{eq:model}
\mathcal{M} \triangleq \langle \mathcal{N}, \mathcal{I}, \mathcal{K}, \mathcal{S}, \mathcal{A},
\mathcal{A}^{+}, \mathcal{O}^{+}, \phi, r, b_0, T \rangle ,
\end{equation}
which keeps $\mathcal{N}$, $\mathcal{I}$, $\mathcal{A} = \{\mathcal{A}_i \mid i \in \mathcal{N}\}$,
$\phi$, $r$, $b_0$ and $T$ from the open Dec-POMDP of Eq.~\eqref{eq:open-decpomdp}, adds the
task set $\mathcal{K}$, reads the state $\mathcal{S}$ as carrying the tasks present, and gives each
agent an availability map $\mathcal{A}_i^{+}$ and an extended observation set
$\mathcal{O}_i^{+}$, collected as $\mathcal{A}^{+} = \{\mathcal{A}_i^{+} \mid i \in \mathcal{N}\}$
and $\mathcal{O}^{+} = \{\mathcal{O}_i^{+} \mid i \in \mathcal{N}\}$.
TO introduces these changes; AO is inherited from the open Dec-POMDP.

\textbf{Task openness.} Tasks do not act, observe, or receive reward, so they need no counterpart to
the coalitions $\mathcal{I}$; we carry them in the state. Every state is a pair:
\begin{equation}
\label{eq:tasks-in-state}
s^t = (x^t, K^t), \qquad K^t \subseteq \mathcal{K},
\end{equation}
where $K^t$ is the set of tasks present at step $t$ and $x^t$ holds the remaining world features,
among them those of the agents and the tasks in $K^t$. The state-reading components keep
their form: $\phi$ makes tasks appear and disappear through the same transition that moves the
state, $r$ may depend on $K^t$, and $b_0$ draws $K^0$ with $s^0$. The tasks present also decide which
actions an agent may take: each agent carries an action availability map
$\mathcal{A}_i^{+} : \mathcal{S} \to 2^{\mathcal{A}_i} \setminus \{\emptyset\}$, and we write
$\mathcal{A}_i^t \triangleq \mathcal{A}_i^{+}(s^t)$ for the actions available to agent $i$ in $s^t$,
so that an action on a task is available only while that task is in $K^t$ and the \textsc{noop}
action, available in every state, keeps $\mathcal{A}_i^t$ nonempty. The joint action is then:
\begin{equation}
\label{eq:available-actions}
\mathbf{a}^t \in \times_{i \in I^t} \mathcal{A}_i^t ,
\end{equation}
which depends on $s^t$, so it never offers actions on absent tasks. A policy
$\pi_\theta(a_i^t \mid \tau_i^t)$ reads only its own history $\tau_i^t$, so it can
select only currently available actions (the support of
Eq.~\eqref{eq:available-actions}) when its observation encodes the set of actions
available in $s^t$. We therefore extend the observation set to
$\mathcal{O}_i^{+} \triangleq \mathcal{O}_i \times (2^{\mathcal{A}_i} \setminus \{\emptyset\})$, and
$\phi$ draws each observation from $\mathcal{O}_i^{+}$; the noise of Section~\ref{sec:bg-perception}
perturbs only the $\mathcal{O}_i$ component, so availability is reported faithfully. Setting
$\mathcal{A}_i^{+}(s) = \mathcal{A}_i$ for all $s$ (equivalently, holding the task
set $K^t$ fixed) removes TO and recovers the open Dec-POMDP.

\textbf{Agent openness.} Agent openness is inherited from the open Dec-POMDP: $\phi$ moves agents of
the population $\mathcal{N}$ into and out of the operating coalition $I^t \in \mathcal{I}$. An agent
outside $I^t$ takes no action and receives no observation, so nothing is appended to its history
while it is away and it resumes that history on return. We restrict the model in one way: no
$\mathcal{A}_i$ contains a \textsc{leave} action, so an agent exits only when its state forces it to,
as when a firefighter's suppressant runs out~\citep{eck2020wildfire,patino2025inaugural}, not by
choice. This is a modeling-scope choice matching our domains, not a requirement
TANDEM's solution relies on.

%% file: sections/4-methodology.tex
We present TANDEM, our training method for TAO-Dec-POMDP. We first introduce the two
reads of its single critic and the shared target that trains them
(Section~\ref{sec:method-critic}), then the parameter-free gate and the policy
update (Section~\ref{sec:method-gate}). We next establish three properties
(Section~\ref{sec:method-theory}), discuss its computational complexity 
(Section~\ref{sec:method-complexity}), and show that it stays well-defined as the
agent and task sets change (Section~\ref{sec:method-invariance}).  TANDEM is a CTDE actor-critic method
(Section~\ref{sec:bg-ctde}) that narrows the information gap without
changing what the policy may read at execution: it keeps a single centralized critic,
reads it two ways, and lets a parameter-free gate decide per agent-step how much of
the centralized read the policy trains on. Figure~\ref{fig:tandem} sketches the method
and Algorithm~\ref{alg:tandem} states it.

\subsection{Two reads of one critic}
\label{sec:method-critic}

Rather than training a second value network, TANDEM reads one centralized critic
two ways. It trains the
decentralized policy $\pi_\theta(a_i^t \mid \tau_i^t)$ and the centralized critic
$V_\psi$ of Section~\ref{sec:bg-mappo}. The actor is a width-64 pre-norm
residual MLP (one residual block) over each agent's own information; we realize $V_\psi$ as an attention
encoder over entity tokens: one token for the global context, one per agent, and
one per task present in $K^t$ (Eq.~\eqref{eq:tasks-in-state}). Attention over
per-entity tokens is a standard construction for variable entity
sets~\citep{hu2021updet,rahman2021gpl,anilmohito} and accepts any number of agents
and tasks (Section~\ref{sec:method-invariance}). The
\emph{centralized (global) read} $V_i^{G,t} = V_\psi(s^t)_i$ scores agent $i$ from the
full token set, where $V_\psi(\cdot)_i$ is the critic's per-agent readout at agent $i$'s
token, unlike MAPPO's single shared $V(s)$. The \emph{decentralized (local) read} $V_i^{W,t} = V_\psi(m_i(s^t))_i$ applies a
mask $m_i$ that hides the tokens agent $i$ does not observe, so $V_i^{W,t}$
depends only on the information agent $i$ will have at execution, namely the
current masked state (a subset of the information in $\tau_i^t$). This is deliberate:
TANDEM targets the state--observation gap, not the memory gap, so the decentralized read is
conditioned on the masked state rather than the full history $\tau_i^t$; unlike history or
history-state critics, it isolates what the global state adds over the local observation.
Both reads come from one batched forward pass (the decentralized read is per agent, so a step uses $1{+}n$ views).

A single regression target trains both reads. From each we form a per-agent
advantage by GAE (Section~\ref{sec:bg-mappo}), $\hat{A}_i^{G,t}$ and
$\hat{A}_i^{W,t}$, under the same rewards. Both reads are regressed
to one target, the centralized return $y_i^t = \hat{A}_i^{G,t} + V_i^{G,t}$,
\begin{equation}
\label{eq:tandem-critic}
L(\psi) = \mathbb{E}_{i,t}\!\left[\bigl(V_i^{G,t} - y_i^t\bigr)^2 +
\bigl(V_i^{W,t} - y_i^t\bigr)^2\right] ,
\end{equation}
where both reads and the target are value-normalized (targets rescaled by running return statistics) as in standard MAPPO.
Because both reads pursue the same target, they come to agree wherever agent $i$'s
information is sufficient, and differ only where it is not.

\subsection{Gated advantage and policy update}
\label{sec:method-gate}

TANDEM decides per agent-step, with a parameter-free gate, how much of the
centralized read to use. When the two reads agree on whether an action did better
or worse than expected, the decentralized read already carries what the agent can
act on at execution; when they disagree, the local view is misleading and the
centralized read should correct it. TANDEM encodes this with a per-agent,
per-step gate:
\begin{equation}
\label{eq:tandem-gate}
\chi_i^t = \tanh\!\left(\frac{\bigl[-\,\hat{A}_i^{G,t}\,\hat{A}_i^{W,t}\bigr]_{+}}
{\mathrm{MAD} + \varepsilon_0}\right),
\end{equation}
where $[\,\cdot\,]_{+} = \max(\cdot,0)$ (the positive part, i.e., a ReLU) and
$\mathrm{MAD}$ is the median absolute deviation (the median of the absolute
deviations from the median) of the products $\hat{A}_j^{G,s}\,\hat{A}_j^{W,s}$ over the active
agent-steps of the batch (agent $i$ is active at step $t$ when $i \in I^t$). The
constant $\varepsilon_0 = 1.0\times10^{-8}$ keeps this denominator positive when
$\mathrm{MAD}=0$, and $\chi_i^t = 0$ for inactive agents. The numerator
$\bigl[-\hat{A}_i^{G,t}\hat{A}_i^{W,t}\bigr]_{+}$ is nonzero exactly when the two
advantages have opposite signs, which we call a \emph{sign-disagreement}: when they
agree in sign the numerator vanishes and the gate is closed ($\chi_i^t = 0$), and on
a sign-disagreement the gate opens ($\chi_i^t > 0$) by an amount that grows with the
size of the disagreement.

The policy is then updated on the mixed advantage:
\begin{equation}
\label{eq:tandem-mix}
\hat{A}_i^t = \chi_i^t\,\hat{A}_i^{G,t} + \bigl(1-\chi_i^t\bigr)\,\hat{A}_i^{W,t},
\end{equation}
through the clipped surrogate of Eq.~\eqref{eq:mappo} with the advantage
$\hat{A}^t$ replaced by $\hat{A}_i^t$, standardized over the active agent-steps as
in standard MAPPO. A closed gate ($\chi_i^t = 0$) thus \emph{withdraws} the
centralized critic on the agent-steps where the decentralized read suffices, and
an open gate ($\chi_i^t > 0$) \emph{reinstates} it where needed. At
its extremes the update is MAPPO's centralized-critic step when the gate is fully
open ($\chi_i^t \to 1$) and a decentralized actor-critic step on
$\hat{A}_i^{W,t}$ when it is closed ($\chi_i^t \equiv 0$); the gate chooses
between them per agent-step, and the shared target
(Equation~\eqref{eq:tandem-critic}) moves the operating point toward the
decentralized end as the reads come to agree.

\begin{algorithm}[t]
\begin{algorithmic}[1]
\Procedure{Train}{environment, actor $\pi_\theta$, critic $V_\psi$}
  \Repeat
    \State roll out $\pi_\theta$: each agent acts \Comment{actor: residual MLP}
    \State $V^{G}, V^{W} \gets \Call{ReadCritic}{s^t}$ \Comment{one critic, two reads}
    \State form advantages $\hat{A}^{G}, \hat{A}^{W}$ by GAE \Comment{same rewards}
    \State $\hat{A} \gets \Call{Gate}{\hat{A}^{G}, \hat{A}^{W}}$ \Comment{mix per agent-step}
    \State update $\theta$ on $\hat{A}$ \Comment{clipped PPO, Eq.~\eqref{eq:mappo}}
    \State fit both reads to the shared target \Comment{Eq.~\eqref{eq:tandem-critic}}
    \State log $\delta$, the sign-disagreement fraction \Comment{diagnostic}
  \Until{environment-step budget is spent}
  \State \Return $\pi_\theta$ \Comment{execution runs the actor alone}
\EndProcedure
\Function{ReadCritic}{$s^t$}
  \State \Return $V_\psi(s^t)$, $V_\psi(m_i(s^t))$ \small{centralized, decentralized read}
\EndFunction
\Function{Gate}{$\hat{A}^{G}, \hat{A}^{W}$}
  \State \Return $\hat{A}$ with $\chi{=}0$ \small{where agree in sign, else open \Comment{Eqs.~\eqref{eq:tandem-gate},~\eqref{eq:tandem-mix}}}
\EndFunction
\end{algorithmic}
\caption{TANDEM training. One centralized critic is read two ways each step; a
parameter-free gate mixes the two advantages per agent-step and a shared target
trains both reads; execution then runs the actor alone.}
\label{alg:tandem}
\end{algorithm}

\begin{figure}[t]
\centering
\includegraphics[width=\columnwidth]{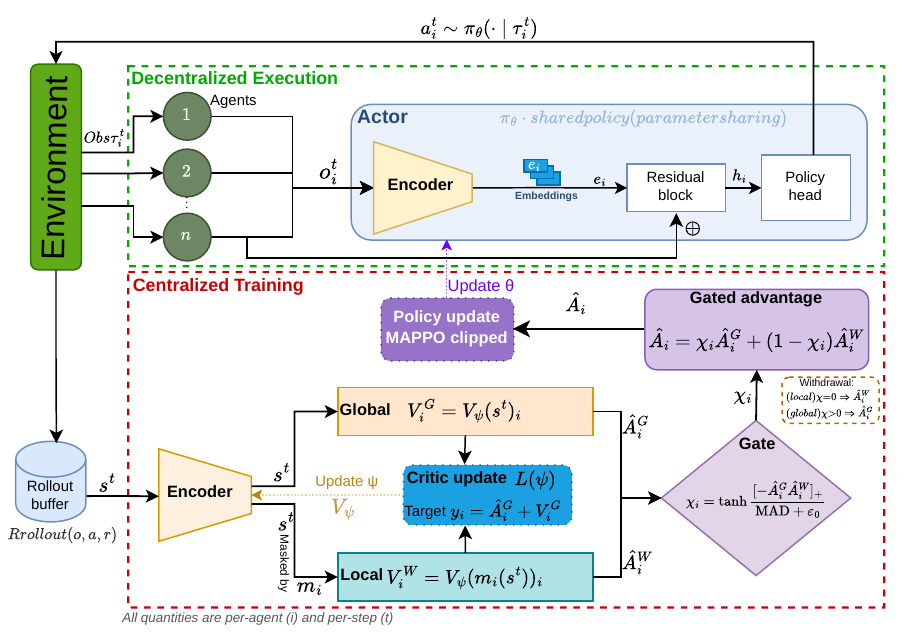}
\caption{TANDEM architecture. One centralized critic is read two ways each step: a
centralized read $V_i^{G,t}$ from the global state $s^t$ and a decentralized read
$V_i^{W,t}$ from the acting agent's information. A shared target trains both reads
(Eq.~\eqref{eq:tandem-critic}); the gate $\chi_i^t$ (Eq.~\eqref{eq:tandem-gate})
mixes their advantages into the policy update (Eq.~\eqref{eq:tandem-mix}).}
\label{fig:tandem}
\end{figure}

\subsection{Theoretical properties}
\label{sec:method-theory}

We establish three properties of TANDEM. First, the information gap degrades the
decentralized policy both at execution and during training
(Proposition~\ref{prop:gap}). Second, TANDEM trains the policy on an advantage
realizable from each agent's own information (Proposition~\ref{prop:fix}). Third,
the gate is scale-free and adds no bias where the two reads agree
(Proposition~\ref{prop:scale}).

\begin{proposition}[the information gap degrades performance]
\label{prop:gap}
Let the policy be restricted to each agent's own information $\tau_i^t$, as
required in Section~\ref{sec:problem}. If the state $s^t$ is not determined
by that information (through partial observability, observation noise, or
openness), then \emph{(i)} at execution no such policy can in general match the
state-conditioned behavior the centralized critic rewards, so its return is
bounded below that benchmark by the value gap of partially observed control; and
\emph{(ii)} during training the centralized advantage $\hat{A}_i^{G,t}$ varies with
state the agent cannot observe and attributes the joint outcome to the agent's
individual action, so on its own it is high-variance and credit-confounded, and it
drives the policy toward the very state-conditioned behavior of \emph{(i)} it
cannot reproduce.
\end{proposition}
\begin{proof}[Proof sketch]
\emph{(i)} Because $s^t$ is not a function of agent
$i$'s information, the best information-restricted policy need not match the
state-conditioned optimum; the difference is the standard value gap between
state-based and information-limited control~\citep{oliehoek2016decpomdp}.
\emph{(ii)} The centralized advantage $\hat{A}_i^{G,t}$ is evaluated from $s^t$
rather than from $\tau_i^t$, so its fluctuations with unobserved state inflate the
variance of the local update. It also credits agent $i$ with transitions driven by
the joint action, an unrealizable, state-conditioned target, so training on it
alone destabilizes learning while steering the policy toward the behavior of
\emph{(i)}.
\end{proof}

\begin{proposition}[TANDEM trains on a realizable advantage]
\label{prop:fix}
Wherever the two advantages agree in sign, the mixed advantage equals the decentralized
advantage ($\chi_i^t = 0 \Rightarrow \hat{A}_i^t = \hat{A}_i^{W,t}$), which is
realizable, in expectation, from agent $i$'s information; the centralized read
enters only on sign-disagreements. The shared target of Equation~\eqref{eq:tandem-critic}
drives $V^{W}$ toward $V^{G}$ on that information, so the two reads differ only up
to the residual information gap of Proposition~\ref{prop:gap}. TANDEM thus restricts
the high-variance centralized advantage to the sign-disagreement steps, where the local
view is misleading, and trains on the realizable decentralized advantage everywhere
else.
\end{proposition}
\begin{proof}[Proof sketch]
The gate of Eq.~\eqref{eq:tandem-gate} is zero iff
$\hat{A}_i^{G,t}\hat{A}_i^{W,t} \ge 0$, giving $\hat{A}_i^t = \hat{A}_i^{W,t}$;
otherwise the gate opens ($\chi_i^t > 0$) and reinstates $\hat{A}_i^{G,t}$.
Minimizing Eq.~\eqref{eq:tandem-critic} drives each read to the conditional
expectation of its (bootstrapped) target, $V^{W} \to \mathbb{E}[y_i^t \mid m_i(s^t)]$
and $V^{G} \to \mathbb{E}[y_i^t \mid s^t]$, which coincide where agent $i$'s information
determines $s^t$ and differ by the information gap elsewhere.
\end{proof}

\begin{proposition}[scale-freedom and unbiasedness at agreement]
\label{prop:scale}
If the two advantages agree in sign on every active agent-step (equivalently the
sign-disagreement fraction $\delta$, the share of active agent-steps where the two
advantages disagree in sign, is $0$), then the gate is closed ($\chi_i^t \equiv 0$)
and the policy (actor) update is exactly a decentralized step on $\hat{A}^{W}$,
adding no bias relative to the decentralized objective. The gate is
also scale-free up to the guard $\varepsilon_0$: for any $c > 0$, replacing
$(\hat{A}^{G}, \hat{A}^{W})$ by $(c\hat{A}^{G}, c\hat{A}^{W})$ leaves $\chi$
unchanged as $\varepsilon_0 \to 0$.
\end{proposition}
\begin{proof}[Proof sketch]
$\delta = 0$ means
$\hat{A}_i^{G,t}\hat{A}_i^{W,t} \ge 0$ for all active $(i,t)$, so the numerator of
Eq.~\eqref{eq:tandem-gate} vanishes and the gate is closed ($\chi \equiv 0$);
Eq.~\eqref{eq:tandem-mix} then gives $\hat{A}_i^t = \hat{A}_i^{W,t}$, so the
policy update of Eq.~\eqref{eq:mappo} is the decentralized step. The critic
still regresses $V^{W}$ to the centralized target $y$ of
Eq.~\eqref{eq:tandem-critic}, so only the actor update is decentralized. For
scale-freedom, write $N_0 = [-\hat{A}_i^{G,t}\hat{A}_i^{W,t}]_{+}$ and
$\mathrm{MAD}_0$ for the unscaled numerator and MAD; both the product and its MAD
scale as $c^2$, so the ratio in Eq.~\eqref{eq:tandem-gate} becomes
$c^2 N_0 / (c^2\,\mathrm{MAD}_0 + \varepsilon_0)$, which tends to
$N_0/\mathrm{MAD}_0$ as $\varepsilon_0 \to 0$.
\end{proof}

\subsection{Computational complexity}
\label{sec:method-complexity}

At execution, TANDEM costs the same as a decentralized actor; during training, it adds a factor of $1{+}n$ on the critic. Let $n$ agents act over a horizon of $T$
steps, with network widths fixed. At execution, it runs only the
actor on each agent's own information, costing $O(nT)$ forward passes and no
critic, as a fully decentralized policy does. During training, a rollout
costs $O(nT)$ actor evaluations, and the critic is read over $1{+}n$ views
per step (the global-state read and one per agent) in a single batched pass over
its at most $1{+}n{+}|\mathcal{K}|$ entity tokens, giving
$O\!\bigl(T\,(1{+}n)\,(1{+}n{+}|\mathcal{K}|)^2\bigr)$ per epoch, the quadratic
factor from attention over the tokens. TANDEM is therefore as scalable as a standard centralized-critic
actor-critic: it adds no new parameters, keeping a shared-parameter
actor and a single centralized, state-value critic~\citep{mappo}, read one
extra way per agent, and changing only the advantage it trains on.
Execution is unchanged from MAPPO, and the method stays well-defined as the agent
and task sets change (Section~\ref{sec:method-invariance}).

\subsection{Invariance to the agent and task sets}
\label{sec:method-invariance}

TANDEM stays well-defined as the sets of present agents and tasks change, by
construction, through its critic and its actor. First, the
critic is an attention encoder over entity tokens, one for the global context, one
per present agent, and one per task in $K^t$, so it accepts any number of agents
and tasks and is equivariant to their ordering. AO and TO change only which tokens
are present, handled by the same masking that defines the two reads.
Second, the actor
$\pi_\theta(a_i^t \mid \tau_i^t)$ is shared across agents and conditioned on an
agent's own information, so it applies unchanged to whichever agents are active
and to the state-dependent action set $\mathcal{A}_i^{t}$ of
Eq.~\eqref{eq:available-actions}. TANDEM therefore needs no fixed agent or
task roster and adds no openness-specific machinery beyond those token masks.

%% file: sections/5-experiments.tex
We evaluate TANDEM against six CTDE
baselines (Section~\ref{sec:bg-ctde}) on nine setups across three cooperative
domains (Figure~\ref{fig:domains}), each isolating one form of the information gap of
Section~\ref{sec:introduction} and one control that removes it. Every method is trained under an identical reward and $1.2$M-step budget and is
measured at execution, where the centralized critic is removed and each agent acts
only on its own $\tau_i^t$ (Section~\ref{sec:bg-ctde}). The supplementary material
lists each setup's gap, agents, tasks, and horizon, with the full variant sweep,
diagnostics, learning curves, network architectures, loss and optimization
details, hyperparameters, the seeding protocol, the per-domain masked reads, and
the compute and software used.

\begin{figure}[t]
\centering
\includegraphics[height=1.75cm]{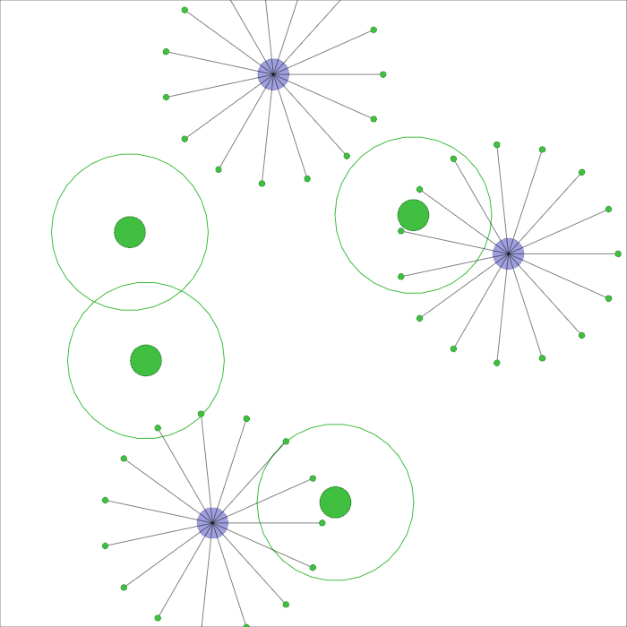}\hfill
\includegraphics[height=1.75cm]{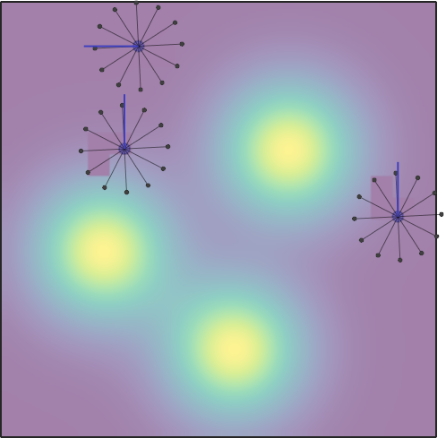}\hfill
\includegraphics[height=1.75cm]{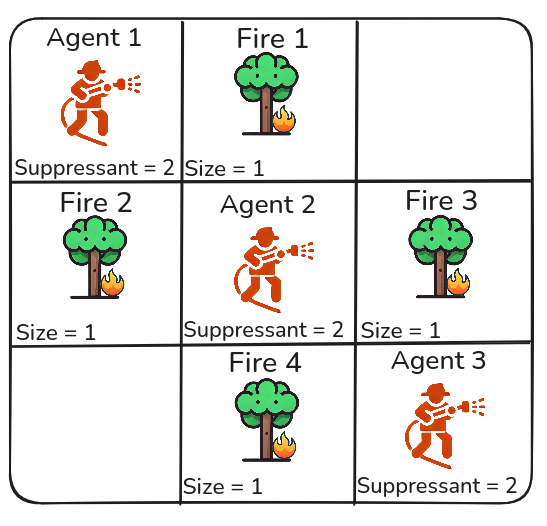}
\caption{The three domains, left to right: discovery, sampling, and wildfire
suppression.}
\label{fig:domains}
\end{figure}

Each domain provides two setups that carry the information gap, a less challenging
one and a more challenging one, together with one control that removes the gap. The
two gap setups raise the difficulty along the domain's axis: a larger team and more
targets in discovery, higher observation noise in sampling, and faster fire spread
and ignition in wildfire suppression. On a control the two critic reads carry the
same information, so the gate stays closed and the actor update is MAPPO's. The
supplementary material gives each setup's agents, tasks, horizon, and full
configuration.

\textbf{Discovery.} A cooperative covering task in
VMAS~\citep{bettini2022vmas}: a team jointly discovers targets, each agent
discovering a target when within its radius and earning reward for every freshly
discovered target. Each agent senses targets only through its short-range LIDAR,
seeing neither absolute positions nor which targets are already discovered, both
held only by the critic's centralized read. Discovery isolates partial
observability: the local view is strictly poorer than the state, with no noise.
The control gives every agent arena-wide LIDAR, removing the partial-observability
restriction.

\textbf{Sampling.} A cooperative information-gathering task in
VMAS~\citep{bettini2022vmas}: a team moves through an arena covered by an
unknown scalar field, a mixture of three Gaussian sources at random
locations, and is rewarded for the field value it samples, each location paying
once. Every step, in training and execution alike, each agent's observation is
perturbed by additive uniform noise and omits a local field view (the nearby
unsampled field, including its direction and remaining amount) that the critic's
centralized read carries. Sampling isolates observation noise: the gap setups
add uniform noise; the control removes it and adds the field view.

\textbf{Wildfire suppression.} A cooperative coordination task on a 3x3
grid~\citep{eck2020wildfire} (a closed and two open fire-rate settings) in
which firefighters act on fires that
occupy cells, each fire a task with a location and intensity. It instantiates  TAO-Dec-POMDP of Section~\ref{sec:problem} with both AO and TO: an agent leaves and rejoins the coalition
as its suppressant runs out and refills, and fires ignite, spread, and are
extinguished within an episode, changing the tasks present. An agent observes its
own position and suppressant, the cells within reach, and a teammate's position,
but not the teammate's suppressant or the per-episode refill probability, both
carried only by the centralized read. The open and more-open gap setups differ in
fire spread and ignition rates. The closed control removes spread, ignition,
and the hidden refill and fully observes teammates. The
team earns a positive reward for each fire extinguished and an equal penalty for each fire that burns out, both scaled by fire size.

\textbf{Baselines.} We compare TANDEM against six CTDE baselines
(Section~\ref{sec:bg-ctde}), grouped by how they use centralized information. The
policy-gradient baselines are: (1) MAPPO~\citep{mappo}, whose policy update TANDEM shares, so
execution return isolates what the two-read gate adds over that step, and (2)
COMA~\citep{foerster2018coma}, which shapes the policy gradient with a
counterfactual baseline from a centralized critic. The value-
factorization
baselines factor a joint action value across agents: (3) QMIX~\citep{rashid2018qmix}, (4) VDN~\citep{sunehag2018vdn}, and
(5) MAVEN~\citep{mahajan2019maven},  whereas (6) MADDPG~\citep{lowe2017maddpg} is an off-policy centralized-critic actor-critic.  All six are standard, well-established CTDE methods that span its policy-gradient, value-factorization, and off-policy centralized-critic families, so the comparison covers how CTDE uses centralized information rather than any single approach.

\textbf{TANDEM configuration.} TANDEM uses the critic, its two reads, the
parameter-free gate, and the shared target exactly as defined in
Sections~\ref{sec:method-critic}--\ref{sec:method-gate}. The actor, the only
component running at execution, conditions on the agent's own observation and a
one-hot identifier. Every method's hyperparameters were tuned with
Optuna~\citep{akiba2019optuna} under the same budget and protocol (supplementary
material), so none is favored by tuning.

\textbf{Training, execution, and statistical protocol.} All methods run 1000 episodes
per run (one policy update per episode over 12 parallel 100-step rollouts) during training, and the
executed policy is the final checkpoint. Seed~42 is only the RNG seed (parameter
initialization and action sampling); it is identical across all methods, so no method
is favored by its draw. The environment seed is distinct from the RNG seed and selects the world
an episode runs in, including its layout, ignitions, and observation noise. During training it is resampled at every episode, so each policy is trained over many environment seeds rather than a single fixed world. Evaluation likewise spans many environment seeds, drawn from a fixed held-out set of
100 worlds (environment seeds $201$--$300$) that is identical for every method and disjoint
from the training and validation seeds, so every evaluation episode is unseen during training.
The evaluation metric is execution return, the mean cumulative team reward per episode
with the critic removed (Section~\ref{sec:bg-ctde}); error bars are its standard error
over the 100 episodes. A difference is \emph{statistically significant} when a paired
two-sided Wilcoxon signed-rank test over these 100 episodes gives $p<0.05$ after Bonferroni correction over the comparisons of a setup, paired because every
method is scored on the identical episodes. Evaluating all methods on the same unseen worlds cancels
the world-to-world difficulty variance that otherwise dominates, so a comparison
reflects the method at a shared RNG seed, not RNG-seed variance.

%% file: sections/6-results.tex
Here we report on TANDEM's execution returns against the six CTDE baselines on the nine
setups (Table~\ref{tab:results}, Figure~\ref{fig:results}).

\textbf{At execution, TANDEM is the best method on 8 of 9 setups}
(Table~\ref{tab:results}) and never significantly below any of the six CTDE
baselines. As Section~\ref{sec:ablations} shows, its widest margins are on the discovery domain, where
the critic's two reads have the widest gap.  

\textbf{TANDEM outperforms all 6 CTDE baselines significantly on 5 of the 9 setups} (Table~\ref{tab:results}): the two discovery gap setups ($3.70$ vs next-best MAPPO
$2.63$, and $6.91$ vs $5.48$), sampling at noise $0.2$ ($37.79$ vs MAPPO $31.68$),
the sampling no-noise control ($41.16$ vs $39.19$), and the wildfire closed control
($48.00$ vs $42.40$). 

\textbf{TANDEM leads on 3 of the 4 remaining setups}: the discovery control ($6.29$
vs MAPPO $5.99$), sampling at noise $0.1$ ($38.87$ vs $37.42$), and the open setup
($290.1$ vs next-best COMA $284.5$). Only on the more-open setup does a baseline's
mean exceed TANDEM's, MAPPO at $315.3$ vs $306.6$; the difference is not
significant (corrected $p=0.68$; Table~\ref{tab:results}), and TANDEM still
outperforms 5 of the 6 baselines there.

\begin{table}[t]
\centering
\scriptsize
\setlength{\tabcolsep}{3pt}
\caption{Execution return (mean over 100 unseen episodes, critic removed, shared
RNG seed~42) of TANDEM and the six CTDE baselines on the nine setups, grouped by
domain (two gap setups and one control each). Bold: best per setup. Superscript
$k$ on TANDEM: the number of the six baselines it outperforms with statistical
significance (Section~\ref{sec:experiments}).}
\label{tab:results}
\resizebox{\columnwidth}{!}{\input{tables/table1_full_results_execution.tex}}
\end{table}

\begin{figure}[t]
\centering
\includegraphics[width=\columnwidth]{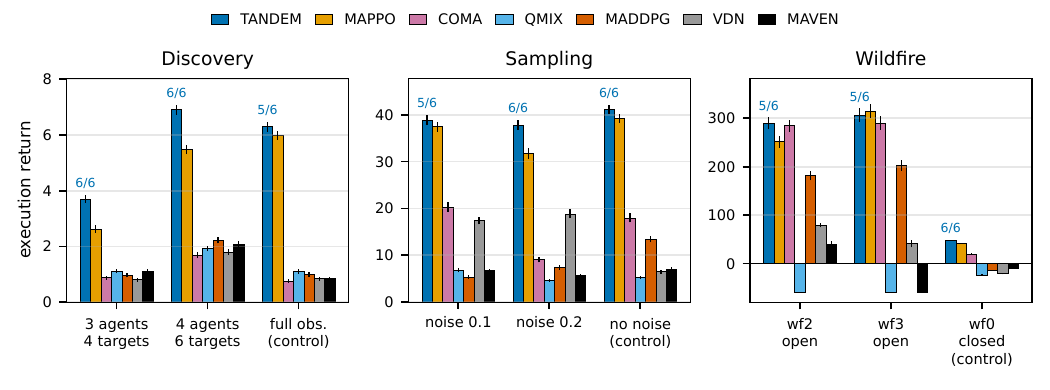}
\caption{Execution return per setup (mean: over 100 unseen episodes; error bars:
standard error; shared RNG seed~42). The label $k/6$ is the
number of the six CTDE baselines TANDEM outperforms with statistical significance (Section~\ref{sec:experiments}).}
\label{fig:results}
\end{figure}

These margins follow the information-gap account of
Section~\ref{sec:method-theory}. The gap between the critic's two reads is widest
on the two discovery gap setups ($0.39$ and $0.41$ in explained variance (EV, the
fraction of return variance the critic predicts); Figure~\ref{fig:gap-disagree}),
as Proposition~\ref{prop:gap} predicts:
a baseline relying entirely or heavily on the centralized critic loses the most once the critic is removed
at execution. Where the reads agree (the sampling no-noise and wildfire
closed controls), the gate stays closed (Section~\ref{sec:method-gate}), yet
its critic still leads.

%% file: tables/table1_full_results_execution.tex
\begin{tabular}{lccccccc}
\toprule
Setup & TANDEM & MAPPO & COMA & QMIX & MADDPG & VDN & MAVEN \\
\midrule
Disc. small & \textbf{3.70$^{6}$} & 2.63 & 0.88 & 1.12 & 0.96 & 0.81 & 1.10 \\
Disc. larger & \textbf{6.91$^{6}$} & 5.48 & 1.68 & 1.93 & 2.23 & 1.79 & 2.07 \\
Disc. control & \textbf{6.29$^{5}$} & 5.99 & 0.77 & 1.10 & 0.99 & 0.84 & 0.85 \\
\midrule
Samp. 0.1 & \textbf{38.87$^{5}$} & 37.42 & 20.34 & 6.81 & 5.37 & 17.38 & 6.79 \\
Samp. 0.2 & \textbf{37.79$^{6}$} & 31.68 & 9.09 & 4.67 & 7.37 & 18.83 & 5.70 \\
Samp. control & \textbf{41.16$^{6}$} & 39.19 & 17.97 & 5.27 & 13.43 & 6.49 & 7.05 \\
\midrule
WF open & \textbf{290.1$^{5}$} & 251.5 & 284.5 & -60.00 & 181.8 & 79.24 & 40.48 \\
WF more open & 306.6$^{5}$ & \textbf{315.3} & 290.5 & -60.00 & 202.5 & 41.60 & -60.00 \\
WF closed (control) & \textbf{48.00$^{6}$} & 42.40 & 19.68 & -23.76 & -13.44 & -20.32 & -9.84 \\
\bottomrule
\end{tabular}

%% file: sections/7-ablations.tex
We isolate what drives TANDEM's performance by removing each of its two
mechanisms in turn, the gate or the shared target
(Table~\ref{tab:ablate-gap}, Figure~\ref{fig:gap-disagree}), while holding the
rest fixed and measuring the change in execution return against the complete TANDEM method (i.e., the \emph{full method})
using the same significance test of Section~\ref{sec:experiments}. 
Both act only on the information gap. Differences that miss
significance are read as not distinguishable rather than ordered.

\textbf{Gate.}
\label{sec:ablation-gate}The gate engages only where a gap exists. For the
full method the gap, gate mean $\bar\chi$, and $\delta$ are nonzero on all
six gap setups and exactly zero on the sampling no-noise and wildfire closed
controls, where the reads agree, the gate stays closed, and every variant
returns the same value (41.16 and 48.00; Table~\ref{tab:ablate-gap}). Across the gap setups the gap is widest on the discovery domain (0.39 and
0.41) and $\delta$ highest on the sampling domain (25.9\% and 25.3\%), so $\delta$ marks the gap's presence,
not its width, across the nine setups (Figure~\ref{fig:gap-disagree}, left;
supplementary material, Table~S3).  Note also that forcing the gate closed ($\chi_i^t \equiv 0$) trains the actor on the
decentralized advantage alone, withdrawing the centralized read on every step
(Proposition~\ref{prop:fix}). Its clearest effect is on the open wildfire
setup, where execution return falls from 290.1 to 272.6
(Table~\ref{tab:ablate-gap}), significantly below the full method (corrected
$p<0.05$; supplementary material, Table~S1): the gate reinstates the centralized read exactly where
openness opens a gap. That is worth the return at execution. On the
other five gap setups, the gate-closed change stays within the per-episode variation
and misses significance (Table~\ref{tab:ablate-gap}).

\begin{table}[t]
\centering
\footnotesize
\setlength{\tabcolsep}{4pt}
\caption{Execution return (mean over 100 unseen episodes) of TANDEM and the
two ablations gate forced closed ($\chi\equiv0$) and no shared target; spread =
max $-$ min of the three per setup (see the supplementary materials for setup details).
Shared RNG seed~42.}
\label{tab:ablate-gap}
\input{tables/table2_ablate_the_gap.tex}
\end{table}

\textbf{Shared target.}
\label{sec:ablation-target}
Removing the shared target lets the decentralized read regress on its own return
instead of the shared target $y_i^t$ (Eq.~\eqref{eq:tandem-critic}). Its
measured effect is at the critic, as designed: on the gap setups it lowers the
decentralized read's fit, its explained variance (EV), and raises $\delta$ (for example sampling noise
0.2 fit 0.77 vs 0.63 and $\delta$ 25.3\% vs 28.7\%; the open setup fit 0.78 vs 0.75
and $\delta$ 18.3\% vs 25.4\%; Table~S2), the convergence of $V^{W}$ toward $V^{G}$ the
target is built to produce (Proposition~\ref{prop:fix}) and the lower $\delta$
across the gap setups in Figure~\ref{fig:gap-disagree} (right). On execution
return, the change is not significant on the gap setups except the larger
discovery map, where removing the shared target scores higher ($p<0.05$; Table~S1), the gate-selectivity pattern
discussed below. Elsewhere, the change is within the per-episode
variation and not significant (the open setup 295.2 and sampling noise 0.1 38.01;
Table~\ref{tab:ablate-gap}).

\begin{figure}[t]
\centering
\includegraphics[width=\columnwidth]{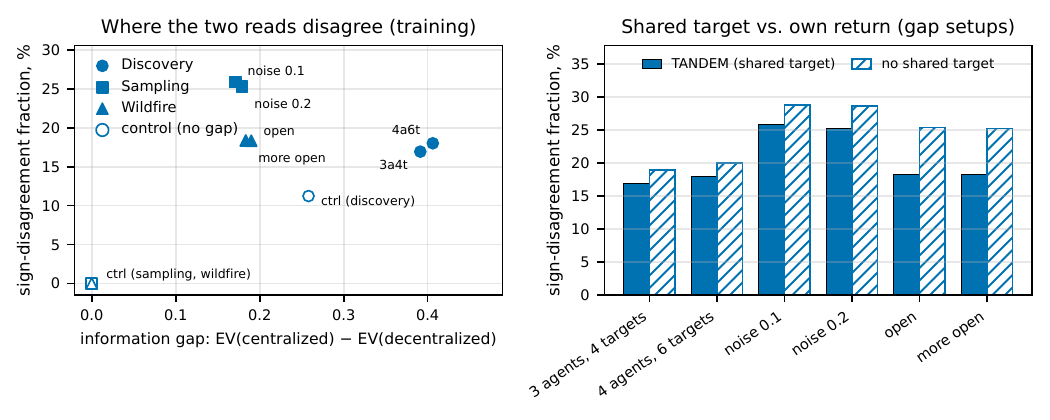}
\caption{Left: sign-disagreement fraction between the centralized and
decentralized advantage against the information gap (EV of the centralized read minus EV
of the decentralized read, last 200 training episodes), one point per setup; marker shape =
domain, hollow = control without a gap. Right: the same fraction with and without
the shared target on the six gap setups (controls are zero for both).}
\label{fig:gap-disagree}
\end{figure}

\textbf{Gate selectivity.}
\label{sec:ablation-select}
The gate is selective to the gap, not the top choice among TANDEM's own
variants on every map. For example, on the larger discovery map (4 agents, 6 targets), three variants score above the full method at corrected $p<0.05$: the gate forced open (9.10), no decentralized read (8.78), and no shared target (7.71), against the full method's 6.91 (Table~S1). There the gate tracks the sign-disagreement signal rather than maximizing return among its variants. The full method's standing rests on its gap-matched gate activity and its never scoring significantly below any CTDE baseline (Section~\ref{sec:results}). On
the other 7 setups no variant differs significantly from the full method.

%% file: tables/table2_ablate_the_gap.tex
\begin{tabular}{lcccc}
\toprule
Setup & TANDEM & \makecell{Gate\\closed} & \makecell{No shared\\target} & Spread \\
\midrule
Disc. small & 3.70 & 3.81 & 3.61 & 0.21 \\
Disc. larger & 6.91 & 7.14 & 7.71 & 0.80 \\
Disc. control & 6.29 & 6.72 & 6.32 & 0.43 \\
\midrule
Samp. 0.1 & 38.87 & 40.32 & 38.01 & 2.32 \\
Samp. 0.2 & 37.79 & 36.71 & 36.75 & 1.08 \\
Samp. control & 41.16 & 41.16 & 41.16 & 0.00 \\
\midrule
WF open & 290.1 & 272.6 & 295.2 & 22.60 \\
WF more open & 306.6 & 309.8 & 301.9 & 7.92 \\
WF closed & 48.00 & 48.00 & 48.00 & 0.00 \\
\bottomrule
\end{tabular}

%% file: sections/9-conclusion.tex
We investigated addressing the information gap in cooperative MARL
under CTDE: a policy trained with a centralized critic is optimized to rely on
global-state information it will not have at execution, leading to a degradation in performance. We proposed TANDEM, which keeps a
single centralized critic, reads it two ways, and mixes the advantages gained from the two reads using a
parameter-free gate that trains the policy on the decentralized advantage where the
two reads agree and on the centralized advantage only where the local view is
misleading.  At execution, with the centralized critic removed, TANDEM is the best method on
8 of the 9 setups and is never significantly worse than any of the 6
CTDE baselines. Its margins and information gap are widest on the discovery
domain, where it significantly outperforms every baseline on the gap setups. This rests on TAO-Dec-POMDP as a formal model of the setting, an instantiation within standard policy-gradient
training, and three theoretical propositions that characterize the method.

TANDEM has limitations that point to future work. It addresses agent and task openness but leaves unexplored
frame openness~\citep{eck2023openness}, in which an agent's goals or capabilities change. The gate acts on the sign agreement of the two reads, so it corrects the
policy where the local view points the wrong way, not where the two views agree yet both are wrong. Its contribution is clearest
under openness, where forcing the gate closed significantly lowers execution
return. On the discovery maps, where the two reads
rarely disagree in sign, the full method already beats every baseline, though its
selective gating is not the highest-scoring among TANDEM's own variants there.
Each domain also isolates a single form of the information gap rather than combining them, and standard cooperative benchmarks are left to future work.
Isolating the gate from the critic architecture through a matched-architecture baseline and extending the two-read construction to frame openness and larger teams are natural next steps.

%% file: sections/10-acknowledgement.tex
This research was supported, in part, by a collaborative NSF Grant
\#IIS-2312657 (to Doshi), \#IIS-2312658 (to Soh), and \#IIS-2312659 (to Eck).
Some of the computing occurred on the Holland Computing Center of the University
of Nebraska, which receives support from the university's Office of Research and
Economic Development and the Nebraska Research Initiative.

%% file: sections/11-supplementary.tex

This supplementary material complements the main paper with the full record behind TANDEM's
result. It collects, in order: the complete variant sweep, with the two ablations
of the main paper
(Section~\ref{sup:variants}); the diagnostic that the shared target does its work at
the critic (Section~\ref{sup:target}); the gate activity measured against the
information gap (Section~\ref{sup:gate}); training-time return against the six CTDE
baselines (Section~\ref{sup:efficiency}); the per-setup learning curves
(Section~\ref{sup:curves}); the training curves of all six CTDE baselines
(Section~\ref{sup:curvesall}); the validation curves during training
(Section~\ref{sup:validation}); the gate dynamics over training
(Section~\ref{sup:dynamics}); the training-versus-execution comparison
(Section~\ref{sup:trainexec}); the hyperparameters (Section~\ref{sup:hparams}); the
per-setup configuration details (Section~\ref{sup:setups}); the seeding protocol in
full (Section~\ref{sup:seeding}); the architecture, optimization, and implementation
details (Section~\ref{sup:arch}); the masked-read contents per domain
(Section~\ref{sup:mask}); and the compute and software statement
(Section~\ref{sup:compute}).

All numbers come from the same experiments as the main paper. Each cell is one trained
policy (RNG seed 42), trained for 1000 episodes under an identical
1.2M-env-step budget and then scored on 100 unseen execution episodes (unseen
environment seeds 201--300, identical for every method) with the centralized critic
removed and each agent acting on its own information. The 100 execution seeds are
disjoint from the training seeds (301 onward) and the validation seeds (100 onward),
so every execution episode is unseen, and each comparison is paired by seed (every
method scored on the identical episodes). Significance throughout is a paired
two-sided Wilcoxon signed-rank test over these 100 paired episodes at $p<0.05$,
Bonferroni-corrected over the comparisons of a setup.
Section~\ref{sup:seeding} states the three seed pools, the RNG seeding, and the
scope of the comparison in full. TANDEM,
the two ablations of the main paper, and the six CTDE baselines are defined in the main paper.

The main paper compares TANDEM against the six CTDE baselines and reports the two
ablations that isolate the gate and the shared target. This supplement adds, for
completeness, the two further TANDEM variants (gate forced open and no decentralized
read), so the full record is available in one place.

\section{All TANDEM variants}
\label{sup:variants}

Table~\ref{tab:variants} widens the comparison to all five TANDEM variants, with
the difference to TANDEM in brackets; an asterisk marks
a difference significant at $p<0.05$ (as defined in the protocol at the start of
this supplement; Bonferroni over seven comparisons per setup, as the correction
also covered three non-CTDE learners omitted from this table). A TANDEM variant
differs significantly from the full method on only two of the nine setups. The main paper uses
two of these variants, the gate forced closed ($\chi_i^t\equiv0$) and no shared
target; the other two, the gate forced open (the policy trained on the centralized
read alone, $\chi_i^t\equiv1$) and no decentralized read (the mask removed, so both
reads see the global state), are reported here for completeness.

The controls behave as the gate predicts. In the sampling no-noise control and the
wildfire closed control setups, all five variants return the identical value (41.16 and
48.00; difference $+0.00$ for every variant), because without the information gap the
two reads agree, the gate stays closed, and the sign-disagreement fraction $\delta$
is $0\%$ (Section~\ref{sup:gate}, Table~\ref{tab:gatepersetup}). The discovery
full-observation control differs only slightly across variants (6.23 to 6.72;
Table~\ref{tab:variants}), since the critic's global-state token and the agent's LIDAR
observation are not the same input even when the arena is fully visible, and no
variant there differs from TANDEM significantly.

The gate adds return under openness. Forcing the gate closed scores 272.6 in the
open wildfire setup, 17.48 below the full method (290.1, $p<0.05$; Table~\ref{tab:variants}),
and does not differ
significantly from TANDEM in any other setup. The gate opens on the open setup precisely
because openness makes the two reads diverge (Section~\ref{sup:gate}), so withdrawing
it there removes the centralized correction the setup needs, while costing nothing on
the setups where the two reads already agree.

On the larger discovery map the gate is selective to the sign-disagreement signal
rather than scoring above its own variants. Forcing the gate open (9.10, $+2.19$,
$p<0.05$), removing the decentralized read (8.78, $+1.87$, $p<0.05$), and removing the
shared target (7.71, $+0.80$, $p<0.05$) each score above the full method (6.91) on
that setup (Table~\ref{tab:variants}); on the smaller discovery map the gate-open and
no-read variants are higher (4.12 and 4.12 vs 3.70) but not significantly, and on the
remaining seven
setups the only significant variant difference is the gate forced closed on the open
wildfire setup ($-17.48$, $p<0.05$; above). On this fully structured
partial-observability map the full method's value is its gap-matched gate activity
(Section~\ref{sup:gate}) together with its never scoring significantly below any CTDE
baseline in the main comparison, rather than a margin over its own ablations.

\begin{table*}[t]
\centering
\footnotesize
\caption{Execution return (mean over 100 unseen episodes) of every TANDEM variant;
in brackets the difference to TANDEM (row minus
TANDEM), $*$ = statistically significant (protocol at the start of this supplement; Bonferroni over
seven comparisons per setup, a correction that also covered three non-CTDE
learners omitted from this table). On
the sampling no-noise and wildfire closed controls every variant returns the same value,
as the gate is closed throughout.}
\label{tab:variants}
\resizebox{\textwidth}{!}{\input{tables/tableS1_all_variants_execution.tex}}
\end{table*}

\section{The shared target at the critic}
\label{sup:target}

Table~\ref{tab:sharedtarget} shows where the shared target does its work: at the
critic, on the decentralized read. For the full method and the no-shared-target
ablation it reports the decentralized read's fit (the explained variance of
$V^W$ against the returns, averaged over the last 200 training episodes) and the
sign-disagreement fraction $\delta$ (over all 1000 training episodes).

The shared target raises the decentralized-read fit and lowers $\delta$ on every gap setup. On
sampling noise 0.2 the decentralized-read fit is 0.77 with the shared target against 0.63
without, and $\delta$ is $25.3\%$ against $28.7\%$; on the open wildfire setup the fit is
0.78 against 0.75 and $\delta$ is $18.3\%$ against $25.4\%$. On the sampling no-noise
and wildfire closed controls $\delta$ is $0.0\%$ with and without the shared target and
the decentralized-read fit is unchanged (0.97 and 0.90), because without the gap the two reads
already agree.

This is the convergence the shared target is designed to produce. Regressing both
reads to one centralized target brings the decentralized read $V^W$ toward the
centralized read $V^G$ on each agent's own information, so the two reads come to agree
wherever that information is sufficient and the gate closes there. The shared target's
measured effect is thus on the critic's two reads, as designed; its effect on
execution return when removed is in Table~\ref{tab:variants}.

\begin{table*}[t]
\centering
\footnotesize
\setlength{\tabcolsep}{16pt}
\renewcommand{\arraystretch}{1.15}
\caption{Fit of the decentralized read (explained variance of $V^W$ over the last 200
training episodes) and sign-disagreement fraction between the two advantages (mean
over 1000 training episodes), with and without the shared target.
The shared target raises the decentralized-read fit and lowers the sign-disagreement fraction on
the six gap setups; on the sampling no-noise and wildfire closed controls the
sign-disagreement fraction is $0\%$ either way.}
\label{tab:sharedtarget}
\input{tables/table3_shared_target_at_the_critic.tex}
\end{table*}

\section{The gate against the information gap}
\label{sup:gate}

Table~\ref{tab:gatepersetup} reports, per setup for the full method, the information
gap (the explained variance of the centralized read minus that of the decentralized
read, over the last 200 of 1000 training episodes), the mean gate value $\bar\chi$ and
the sign-disagreement fraction $\delta$ (over all 1000 episodes), and the environment
steps to reach 90\% of the final training return for the full method and the
gate-forced-closed variant.

The gate engages only where a gap exists, and more where it is wider. The gap, the
gate, and $\delta$ are nonzero on all six gap setups and exactly zero on the sampling
no-noise and wildfire closed controls, and small on the discovery full-observation
control (gap 0.26, gate 0.019, $\delta$ 11.2\%). The gap is largest on discovery
(0.39 and 0.41) while $\delta$ is highest on sampling (25.9\% and 25.3\%): $\delta$
marks the gap's presence, not its width.

\begin{table*}[t]
\centering
\footnotesize
\setlength{\tabcolsep}{4pt}
\caption{The gate against the information gap, full method: gap = EV(centralized read $V^G$) $-$
EV(decentralized read $V^W$) over the last 200 training episodes; gate mean $\bar\chi$ and
sign-disagreement fraction over the 1000 episodes; environment steps until the running
mean (20 episodes = 24k steps) of the training return first reaches 90\% of its final
rise. The gap, the gate, and the sign-disagreement fraction are
exactly zero on the sampling no-noise and wildfire closed controls and small on the
discovery full-observation control.}
\label{tab:gatepersetup}
\resizebox{\textwidth}{!}{\input{tables/table4_gate_against_the_gap.tex}}
\end{table*}

\section{Training efficiency}
\label{sup:efficiency}

Table~\ref{tab:efficiency} reports the mean training return over the 1000 training
episodes, the area under the learning curve, for TANDEM and the six CTDE baselines.
Under the identical 1.2M-env-step budget TANDEM's mean training return is the highest
of the seven methods on every setup, and exceeds that of MAPPO from the larger
discovery map (6.44 vs 4.16) to the open wildfire setup
(259.7 vs 234.1).

\begin{table*}[t]
\centering
\footnotesize
\caption{Mean training return over the 1000 training episodes (area under the learning
curve), TANDEM and the six CTDE baselines. TANDEM's mean training
return is the highest of the seven methods on every setup.}
\label{tab:efficiency}
\resizebox{\textwidth}{!}{\input{tables/table5_training_efficiency.tex}}
\end{table*}

\section{Learning curves}
\label{sup:curves}

Figure~\ref{fig:curves} shows the 20-episode running mean of the training return for
TANDEM and its two ablations, the gate forced closed and no shared target, with
MAPPO as a reference, one panel per setup. TANDEM tracks at or above the ablations
and MAPPO throughout training, and the ablation curves separate from TANDEM mainly
where the information gap drives the gate open, most visibly on the open setup.

\begin{figure*}[t]
\centering
\includegraphics[width=\textwidth]{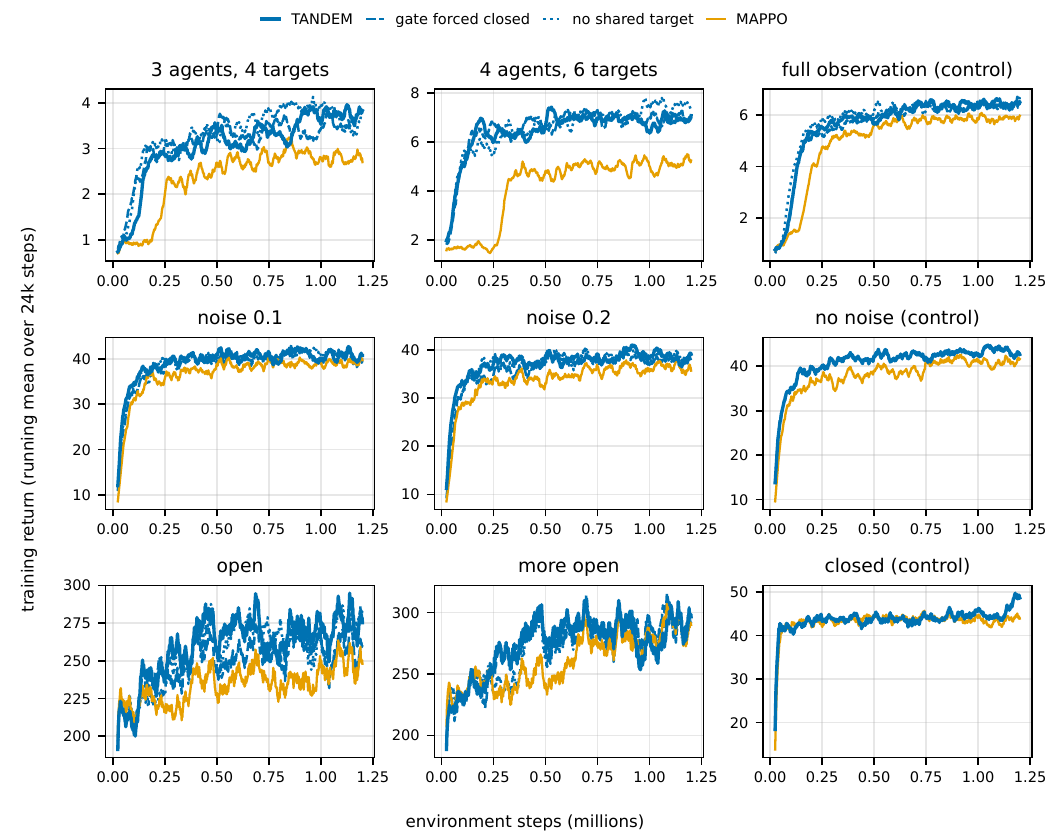}
\caption{Training return against environment steps (curves smoothed with a running mean over 24k steps = 20
episodes) of TANDEM and its two ablations, gate forced closed ($\chi\equiv0$) and no
shared target, with MAPPO as reference, one panel per setup; no
band.}
\label{fig:curves}
\end{figure*}

\section{Training curves for all CTDE baselines}
\label{sup:curvesall}

Figure~\ref{fig:curvesall} shows the 20-episode running mean of the training return
for TANDEM and all six CTDE baselines, one panel per setup. Consistent with
Table~\ref{tab:efficiency}, TANDEM reaches the highest training return of the seven
methods on every setup. The policy-gradient methods TANDEM and MAPPO lead the
value-factorization baselines (QMIX, VDN, MAVEN) and MADDPG by a wide margin on
wildfire, where the open-setup return is hardest to reach.

\begin{figure*}[t]
\centering
\includegraphics[width=0.8\textwidth]{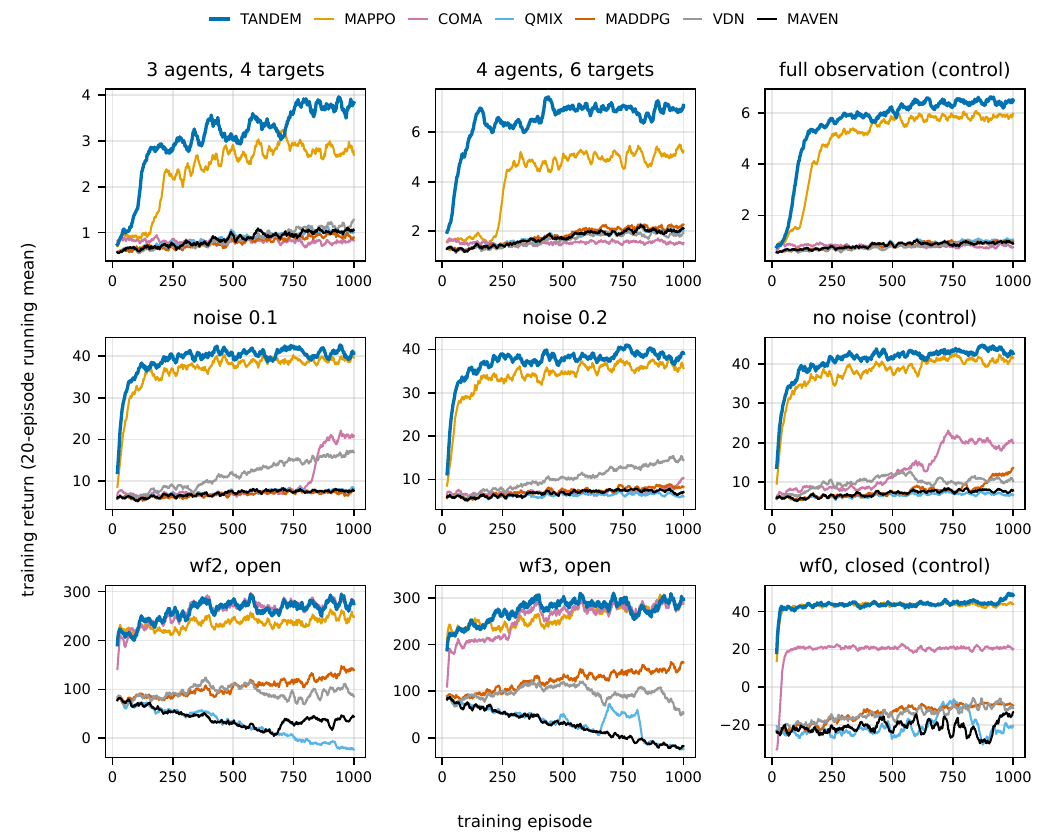}
\caption{Training return over the full 1000 training episodes (1000 episodes of 12
parallel environments $\times$ 100 steps = 1.2M steps) of TANDEM and the six CTDE
baselines, one panel per setup; curves smoothed with a 20-episode running mean
(24k environment steps); no band.}
\label{fig:curvesall}
\end{figure*}

\section{Validation curves during training}
\label{sup:validation}

Figure~\ref{fig:validation} shows the validation return during training for TANDEM
and the six CTDE baselines, one panel per setup. Validation is measured every 40
episodes on the unseen validation seeds (100 onward; 30 validation episodes for
discovery and sampling, 50 for wildfire), seeds disjoint from both the training seeds
(301 onward) and the execution seeds (201--300), so the curves track generalization to
unseen environments during training without being used for checkpoint selection.
TANDEM's validation return is at or near the top of the seven methods throughout
training on each setup, consistent with its execution result on the unseen seeds.

\begin{figure*}[t]
\centering
\includegraphics[width=\textwidth]{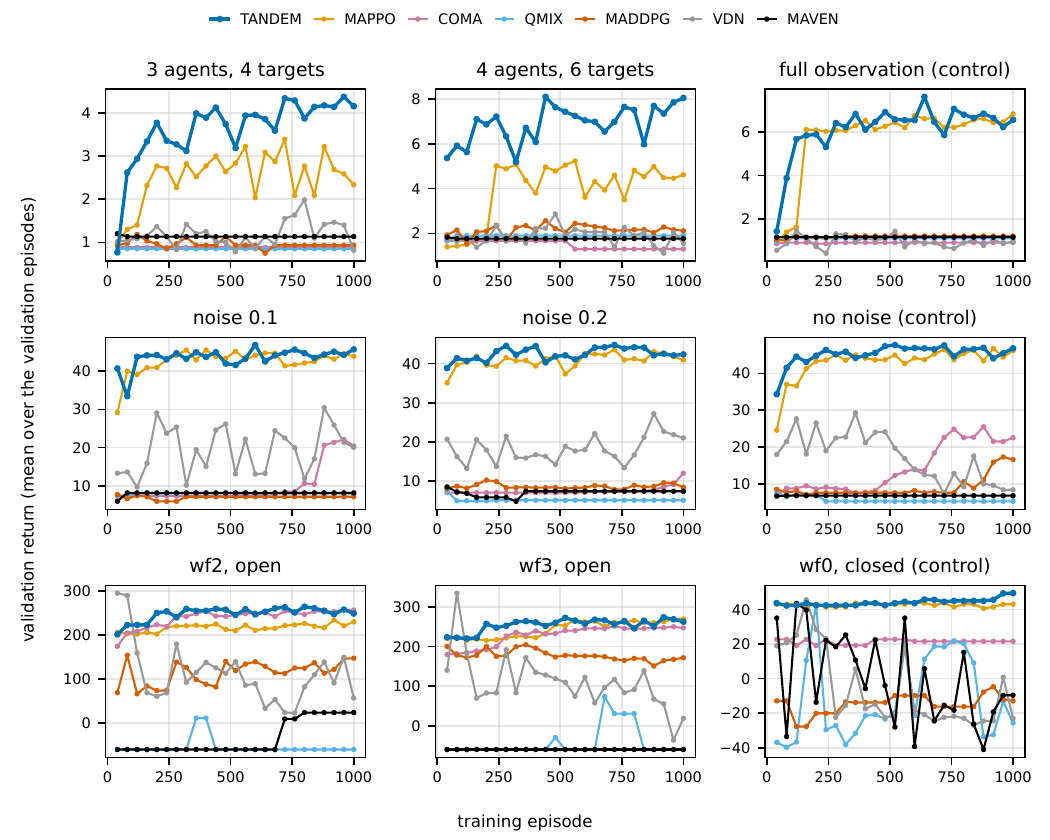}
\caption{Validation return against training episodes (every 40 episodes = 48k environment steps;
mean over the validation episodes on the unseen validation seeds 100 onward; 30 episodes
for discovery and sampling, 50 for wildfire), TANDEM and the six CTDE baselines, one
panel per setup.}
\label{fig:validation}
\end{figure*}

\section{Gate dynamics over training}
\label{sup:dynamics}

Figure~\ref{fig:gatedyn} shows the gate mean $\bar\chi$ and the sign-disagreement
fraction $\delta$ over training, one panel per setup. Both are sustained and nonzero
on the six gap setups and flat at zero on the sampling no-noise and wildfire closed
controls; on the discovery full-observation control they rise slowly, because the
critic's global-state token and the agent's LIDAR observation remain different inputs.
The gate therefore stays engaged exactly where the information gap is present,
throughout training, not only at the end.

\begin{figure*}[t]
\centering
\includegraphics[width=\textwidth]{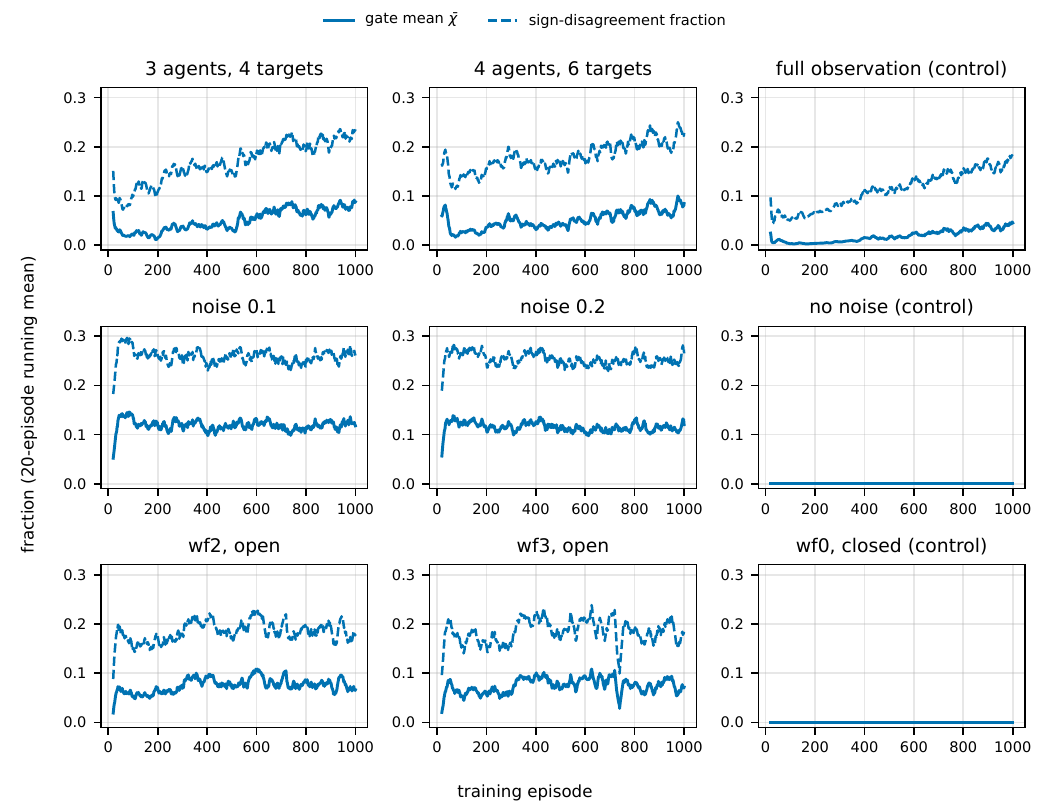}
\caption{TANDEM's gate mean $\bar\chi$ and sign-disagreement fraction against
training episodes (smoothed with a 20-episode running mean), per setup: both stay at zero on the
controls without a gap.}
\label{fig:gatedyn}
\end{figure*}

\section{Training versus execution return}
\label{sup:trainexec}

Figure~\ref{fig:trainexec} places each method's training return (over steps
1.08M--1.2M, the last 100 training episodes, with the critic shaping the policy) next to its execution return (the 100
unseen episodes, with the critic removed), one panel per setup. The information gap
is the drop from training to execution: a method optimized against a centralized
critic can score well in training and then lose return once the critic is gone.
TANDEM sits at or near the top in training and holds that return at execution, with a
small train-to-execution change relative to the baselines that lean hardest on a
centralized signal.

\begin{figure*}[t]
\centering
\includegraphics[width=\textwidth]{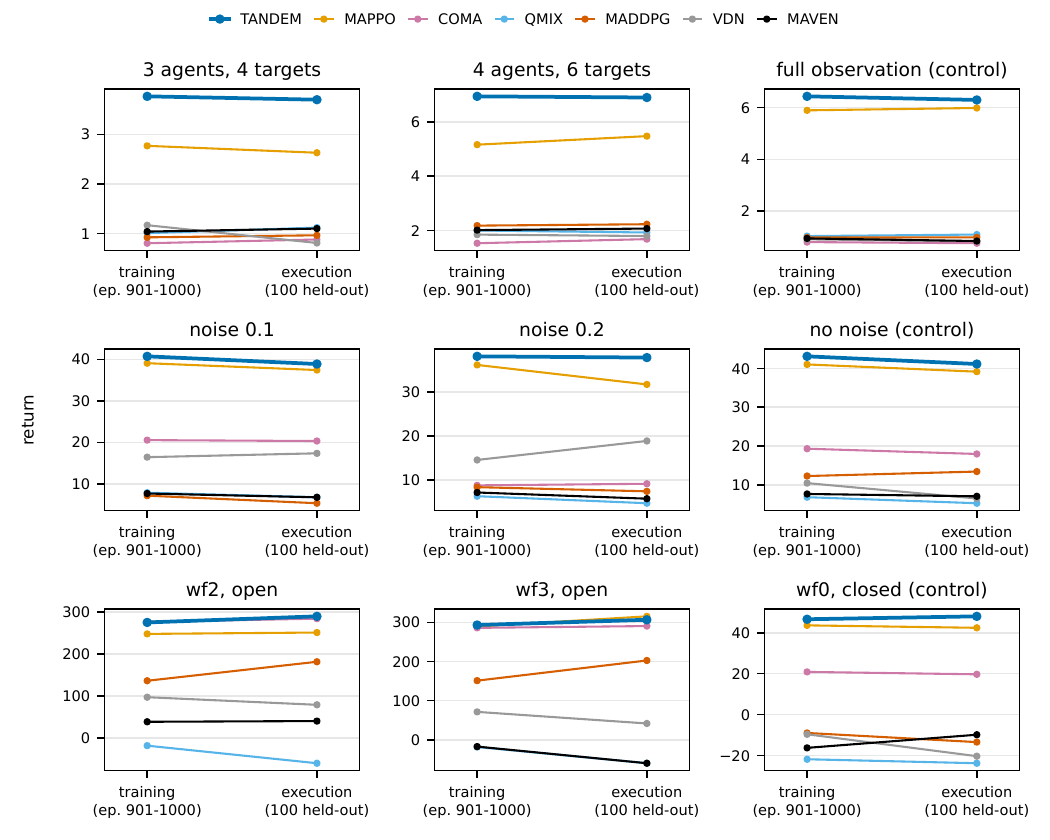}
\caption{Mean training return over the last 100 training episodes (steps 1.08M--1.2M,
left point) joined to the mean execution return over 100 unseen episodes (right
point), per method and setup.}
\label{fig:trainexec}
\end{figure*}

\clearpage
\section{Hyperparameters}
\label{sup:hparams}

\begin{sloppypar}
All methods---TANDEM and the baselines alike---had their hyperparameters tuned with
Optuna under one shared protocol and budget, so no method is
favored by tuning. Each trial trained a configuration on all three domains at a reduced
probe budget (100 episodes per domain for the on-policy methods, 300 for the
off-policy baselines, whose replay warm-up and $\epsilon$ anneal need the longer
probe),
with a fixed trial RNG seed so trials are comparable, and
was scored by the mean per-domain normalized evaluation return on the validation seed
pool (100 onward), disjoint by construction from the held-out execution seeds
201--300 used for every reported number, so tuning never saw the test worlds; a TPE
sampler proposed configurations and a median pruner stopped trials that fell behind
after a domain, run as one resumable study per method with a matched per-study budget.
The selected configuration (highest mean normalized return) is the one used in every
reported run and listed in Tables~\ref{tab:hparams} and~\ref{tab:hparams-value}. Search
ranges followed best practice per algorithm family: for the policy-gradient methods,
learning rate $[10^{-4},3{\times}10^{-3}]$ (log), entropy coefficient
$[10^{-4},5{\times}10^{-2}]$ (log), GAE $\lambda\in[0.90,0.99]$, clip $\in[0.10,0.30]$,
PPO epochs $\in[5,15]$, value-loss coefficient $[0.25,1.0]$, gradient clip $[0.5,10]$
(log), hidden width $\in\{64,128,256\}$ and depth $\in\{1,2,3\}$; for the value-based
methods, learning rate, batch size, buffer size, target-update interval and the
$\epsilon$ schedule over comparable ranges; for MADDPG, separate actor and critic
rates, $\tau$, batch and buffer. All methods share the fixed budget of 1000 episodes
($1.2$M env steps) over 12 parallel 100-step rollouts and discount $\gamma=0.99$;
the policy-gradient methods and MADDPG use Adam ($\varepsilon=10^{-5}$), the
discrete value-based methods RMSprop ($\alpha=0.99$, $\varepsilon=10^{-5}$).
\end{sloppypar}

Three notes on the value-based methods (Table~\ref{tab:hparams-value}). First,
$\epsilon$ anneals linearly per training episode over 2000 episodes, the production
default shared across the value-based methods; under the 1000-episode budget the
final checkpoint therefore collects experience at $\epsilon=0.525$, while validation
and execution are always greedy ($\epsilon=0$). Second, MAVEN uses a latent
dimension $|z|=8$ and a DIAYN-style mutual-information intrinsic reward with
coefficient $0.1$; its hierarchical $z$-policy and per-state MI discriminator are
2-hidden-layer width-64 Tanh MLPs, and the $z$-policy is trained by REINFORCE
(running baseline) with its own Adam at $5{\times}10^{-4}$. Third, the $-60.00$ entries on the
two open wildfire setups (3$\times$3 grid, hidden refill rate) are the deterministic
value of a degenerate never-extinguish policy: when all six fires burn out, the
penalty is exactly $-(3\cdot4+3\cdot16)=-60$, so $-60.00$ marks that outcome, not an
evaluation artifact.

\begin{table}[t]
\centering
\footnotesize
\setlength{\tabcolsep}{4pt}
\caption{Tuned hyperparameters: policy-gradient methods. The actor is a width-64 MLP (2 layers), or a width-64
pre-norm residual block (1 block) for TANDEM; COMA's and TANDEM's critics are
transformers over entity tokens ($d_{\text{model}}{=}128$, 2 layers). TANDEM clips
actor and critic gradients separately; its per-read MSE value loss is unweighted
(no value-loss coefficient, marked ---).}
\label{tab:hparams}
\begin{tabular}{lccc}
\toprule
 & MAPPO & COMA & TANDEM \\
\midrule
Learning rate    & $5{\times}10^{-4}$ & $5{\times}10^{-4}$ & $5{\times}10^{-4}$ \\
Entropy coef.    & $0.01$ & $0.01$ & $0.01$ \\
GAE $\lambda$    & $0.95$ & $0.95$ & $0.95$ \\
Clip ratio       & $0.2$ & $0.2$ & $0.2$ \\
PPO epochs       & $10$ & $10$ & $10$ \\
Value-loss coef. & $0.5$ & $0.5$ & --- \\
Grad.\ clip      & $10$ & $10$ & $0.5$ / $10$ \\
Hidden width     & $64$ & $64$ & $64$ \\
Hidden layers    & $2$ & $2$ & $1$ \\
\bottomrule
\end{tabular}
\end{table}

\begin{table}[t]
\centering
\footnotesize
\setlength{\tabcolsep}{4pt}
\caption{Tuned hyperparameters: value-based methods and MADDPG. Hidden width 64 and 2
layers throughout; mixing-network width 32 for QMIX and MAVEN. $\epsilon$ anneals
$1.0\!\to\!\epsilon_{\text{end}}$ over the listed episodes; MADDPG lists its
actor/critic learning rates and uses a soft target update ($\tau$) rather than a hard
interval. After its warm-up, each method takes one gradient step per environment
step; MADDPG's collection uses Gumbel-softmax at temperature 1.0 with the
$\epsilon$ overlay.}
\label{tab:hparams-value}
\begin{tabular}{lcccc}
\toprule
 & QMIX & VDN & MAVEN & MADDPG \\
\midrule
Learning rate           & $5{\times}10^{-4}$ & $5{\times}10^{-4}$ & $5{\times}10^{-4}$ & $10^{-4}/10^{-3}$ \\
Batch size              & $32$ & $32$ & $32$ & $1024$ \\
Buffer size             & $5000$ & $5000$ & $5000$ & $10^{6}$ \\
Target update           & $10$ & $10$ & $10$ & $\tau{=}0.01$ \\
$\epsilon_{\text{end}}$ & $0.05$ & $0.05$ & $0.05$ & $0.05$ \\
$\epsilon$-decay (ep.)  & $2000$ & $2000$ & $2000$ & $2000$ \\
Grad.\ clip             & $10$ & $10$ & $10$ & $10$ \\
Warm-up (env steps)     & $500$ & $500$ & $500$ & $1000$ \\
\bottomrule
\end{tabular}
\end{table}

\section{Setup details}
\label{sup:setups}

Table~\ref{tab:setupdetail} gives each setup's information gap, the agent and task
counts at episode start, and the horizon. Each domain has two setups that carry the
gap, a less challenging one and a more challenging one, and one control that removes
it; on a control the two critic reads carry the same information (exactly so on
the sampling and wildfire controls, where every TANDEM variant coincides; nearly
so on the discovery control, where the masked read still omits teammate tokens),
so the gate is closed or nearly closed and contributes essentially nothing, and
the remaining difference to MAPPO reflects the critic architecture.

\begin{table}[t]
\centering
\footnotesize
\setlength{\tabcolsep}{4pt}
\begin{tabular}{llccc}
\toprule
Setup & Information gap & Agents & Tasks & Horizon \\
\midrule
\multicolumn{5}{l}{\emph{Discovery} (VMAS)} \\
\quad 3 agents, 4 targets & Partial observability & 3 & 4 & 100 \\
\quad 4 agents, 6 targets & Partial observability & 4 & 6 & 100 \\
\quad full observation (control) & none & 3 & 4 & 100 \\
\addlinespace
\multicolumn{5}{l}{\emph{Sampling} (VMAS)} \\
\quad noise $0.1$ & Observation noise & 3 & 3 & 100 \\
\quad noise $0.2$ & Observation noise & 3 & 3 & 100 \\
\quad no noise (control) & none & 3 & 3 & 100 \\
\addlinespace
\multicolumn{5}{l}{\emph{Wildfire suppression}} \\
\quad open & Openness & 3 & 2 & 100 \\
\quad more open & Openness & 3 & 2 & 100 \\
\quad closed (control) & none & 3 & 6 & 100 \\
\bottomrule
\end{tabular}
\caption{The nine setups and the gap each isolates. Counts are at episode start;
each control removes its domain's gap. Wildfire's agent and task sets change within
an episode, as described in the main paper.}
\label{tab:setupdetail}
\end{table}

\section{Seeding protocol in full}
\label{sup:seeding}

Every run draws from three disjoint, method-independent, run-independent seed
pools. \emph{Training seeds} run from 301 onward: episode $e$, parallel environment
$j$ uses environment seed $301+(e-1)\cdot12+j$, a pure function of the episode
index, never of the method or the RNG seed, so all methods see byte-identical
training worlds, roughly 12{,}000 distinct worlds over the 1000 episodes of 12
parallel environments. \emph{Validation seeds} run from 100 onward (30 episodes; 50
on wildfire) and are used for the validation curves only, never for checkpoint
selection. \emph{Execution seeds} are 201--300, the held-out set behind every
reported number.

The RNG seed is separate: both configuration knobs are set to 42 and mixed into one
process seed that seeds Python, NumPy, and PyTorch, with
\texttt{torch.use\_deterministic\_algorithms} enabled. It controls parameter
initialization, action sampling, and minibatch order, and is fixed ex ante and
shared across all methods.

Scope of the comparison. Each table cell is one trained policy per method at this
shared draw; the Wilcoxon test over the 100 paired held-out episodes tests
differences between these policies on unseen worlds. The design trades replication
over RNG seeds for breadth over environment seeds under a fixed compute budget, so
comparisons are reported conditional on the shared draw.

\section{Architecture, optimization, and implementation details}
\label{sup:arch}

\textbf{TANDEM actor.} A width-64 pre-norm residual MLP with one residual block
(LayerNorm, Linear, GELU, Linear, with a skip connection) and a final LayerNorm;
GELU activations throughout; orthogonal initialization (gain $\sqrt{2}$ for hidden
layers, $0.01$ for the output layer). Its input is the raw observation concatenated
with a one-hot agent identifier, and its parameters are shared across agents.

\textbf{TANDEM critic.} A 2-layer pre-norm transformer encoder over entity tokens:
$d_{\text{model}}=128$, 4 attention heads, feed-forward width 256, GELU, no
dropout. Each token type (global, agent, task) has its own input projection; a
final LayerNorm and a linear value head (gain $1.0$) produce the value, read out
per agent at the agent's own token. Masked tokens are excluded as attention keys,
and both reads are computed in one batched pass.

\textbf{TANDEM losses and optimization.} Adam (learning rate $5\times10^{-4}$,
$\varepsilon=10^{-5}$) for actor and critic; 10 PPO epochs on the full batch (no
minibatching). The actor loss is the clipped PPO surrogate of the main paper on
the mixed advantage $\hat{A}_i^t$ (clip $0.2$), plus an entropy bonus with
coefficient $0.01$, averaged over the active (death-masked) agent-steps, with
advantages standardized per batch; the actor gradient is clipped at $0.5$. The
gate $\chi$ is evaluated once per collected batch, under no-grad, from the two GAE
advantages in reward units (both reads are de-normalized by the running
value-normalizer before GAE), so no gradient flows through the gate and $\chi$
stays fixed across the epochs; the MAD normalizer is taken over the active
agent-steps of the batch, with guard $\varepsilon_0=10^{-8}$. The critic loss is
the unweighted sum of the per-read MSEs of $V^{G}$ and $V^{W}$ against the single
shared normalized target, the normalized centralized return
$y_i^t=\hat{A}_i^{G,t}+V_i^{G,t}$; one running value-normalizer (an EMA of the
return mean and variance, $\beta=0.99999$) is refreshed once per update; the
critic gradient is clipped at $10.0$. GAE uses $\gamma=0.99$ and $\lambda=0.95$,
computed per read, with bootstrapped values at time-limit truncation (the
100-step horizon), so episode cuts are not treated as terminal.

\textbf{Baselines.} MAPPO's actor and critic are 2-hidden-layer
width-64 MLPs with Tanh activations. COMA uses MAPPO's actor and the same
entity-token transformer critic as TANDEM ($d_{\text{model}}=128$, 2 layers, 4
heads) with a counterfactual Q head. QMIX and VDN use per-agent Q-networks
(observation $\to$ 64 $\to$ 64 $\to$ actions, ReLU) trained with RMSprop (learning
rate $5\times10^{-4}$, $\alpha=0.99$), batch size 32, buffer size 5000, and target
networks synced every 10th episode; the QMIX and MAVEN mixing networks have width
32. MADDPG uses 64$\times$2 ReLU MLPs, actor learning rate $10^{-4}$ and critic
learning rate $10^{-3}$, soft target update $\tau=0.01$, batch size 1024, and a
$10^{6}$ buffer. All baselines use gradient clip $10.0$.

\section{Masked-read contents per domain}
\label{sup:mask}

This section lists what the masked state $m_i(s^t)$ of the decentralized read
conditions on in each domain.

\textbf{Discovery} masks by token removal: $m_i(s^t)$ retains the global token,
agent $i$'s own token, and the targets within LIDAR range; teammates are never
shown. This domain has no observation noise, so the masked state matches the
agent's own information.

\textbf{Sampling} masks by token removal: $m_i(s^t)$ retains the global token, the
agent's own token, and the teammates within LIDAR range; the Gaussian-field tokens
are never shown, and the critic-only columns (the hidden noise level and the
field-view summaries) are zeroed in the agent's read. One caveat holds here: the
retained tokens carry noise-free state features (own position and velocity, visible
teammates' positions), while the actor's observation of the same quantities carries
additive uniform noise; under observation noise the decentralized read is therefore
a denoised view of the entities agent $i$ observes, not a strict subset of the
actor's information.

\textbf{Wildfire} masks by value-level redaction: teammates' water and power are
set to $-1$ and their range, capacity, and gear are zeroed; the hidden per-episode
refill rate is zeroed in the agent's read; under noise, teammates' water and power
and each cell's lit/intensity readings are replaced by agent $i$'s own stored noisy
view (range, capacity, and gear stay zeroed); the agent's own token is never
overridden.

\section{Compute and software}
\label{sup:compute}

All final-campaign training ran CPU-only on a Slurm cluster, each training job
using 1 task with 2 CPU cores and 16\,GB of memory under a 72-hour wall limit: 126
training jobs (9 setups $\times$ 14 training configurations, including runs not
reported here), plus 126 held-out
execution jobs (2 cores, 8\,GB). Software: PyTorch $\geq$ 2.6 with deterministic
algorithms enabled; key packages gymnasium 1.0.0, pettingzoo 1.24.3, vmas 1.5.2,
and free-range-zoo 0.1.24.

%% file: tables/tableS1_all_variants_execution.tex
\begin{tabular}{lccccccccc}
\toprule
Variant & \makecell{Discovery\\small} & \makecell{Discovery\\larger} & \makecell{Discovery\\control} & \makecell{Sampling\\noise 0.1} & \makecell{Sampling\\noise 0.2} & \makecell{Sampling\\control} & \makecell{Wildfire\\open} & \makecell{Wildfire\\more open} & \makecell{Wildfire\\closed (control)} \\
\midrule
TANDEM & 3.70 & 6.91 & 6.29 & 38.87 & 37.79 & 41.16 & 290.1 & 306.6 & 48.00 \\
gate forced open & 4.12 (+0.42) & 9.10 (+2.19$^{*}$) & 6.67 (+0.37) & 40.21 (+1.34) & 36.92 (-0.87) & 41.16 (+0.00) & 288.0 (-2.08) & 298.0 (-8.56) & 48.00 (+0.00) \\
gate forced closed & 3.81 (+0.11) & 7.14 (+0.23) & 6.72 (+0.43) & 40.32 (+1.45) & 36.71 (-1.08) & 41.16 (+0.00) & 272.6 (-17.48$^{*}$) & 309.8 (+3.24) & 48.00 (+0.00) \\
no decentralized read & 4.12 (+0.42) & 8.78 (+1.87$^{*}$) & 6.23 (-0.06) & 40.57 (+1.70) & 36.41 (-1.38) & 41.16 (+0.00) & 287.4 (-2.68) & 312.7 (+6.16) & 48.00 (+0.00) \\
no shared target & 3.61 (-0.09) & 7.71 (+0.80$^{*}$) & 6.32 (+0.03) & 38.01 (-0.87) & 36.75 (-1.05) & 41.16 (+0.00) & 295.2 (+5.12) & 301.9 (-4.68) & 48.00 (+0.00) \\
\bottomrule
\end{tabular}

%% file: tables/table3_shared_target_at_the_critic.tex
\begin{tabular}{lcccc}
\toprule
& \multicolumn{2}{c}{Decentralized-read fit} & \multicolumn{2}{c}{Sign-disagreement} \\
\cmidrule(lr){2-3}\cmidrule(lr){4-5}
Setup & shared & no shared & shared & no shared \\
\midrule
Disc. small & 0.06 & -0.01 & 16.9\% & 18.9\% \\
Disc. larger & 0.04 & 0.03 & 18.0\% & 20.0\% \\
Disc. control & 0.21 & 0.21 & 11.2\% & 14.6\% \\
\midrule
Samp. 0.1 & 0.79 & 0.73 & 25.9\% & 28.8\% \\
Samp. 0.2 & 0.77 & 0.63 & 25.3\% & 28.7\% \\
Samp. control & 0.97 & 0.97 & 0.0\% & 0.0\% \\
\midrule
WF open & 0.78 & 0.75 & 18.3\% & 25.4\% \\
WF more open & 0.79 & 0.76 & 18.4\% & 25.2\% \\
WF closed (control) & 0.90 & 0.90 & 0.0\% & 0.0\% \\
\bottomrule
\end{tabular}

%% file: tables/table4_gate_against_the_gap.tex
\begin{tabular}{lccccccccc}
\toprule
Reading & \makecell{Discovery\\small} & \makecell{Discovery\\larger} & \makecell{Discovery\\control} & \makecell{Sampling\\noise 0.1} & \makecell{Sampling\\noise 0.2} & \makecell{Sampling\\control} & \makecell{Wildfire\\open} & \makecell{Wildfire\\more open} & \makecell{Wildfire\\closed (control)} \\
\midrule
gap: EV($V^G$) $-$ EV($V^W$) & 0.39 & 0.41 & 0.26 & 0.17 & 0.18 & 0.00 & 0.19 & 0.18 & 0.00 \\
gate mean $\bar\chi$ & 0.051 & 0.051 & 0.019 & 0.118 & 0.115 & 0.000 & 0.075 & 0.074 & 0.000 \\
sign-disagreement & 16.9\% & 18.0\% & 11.2\% & 25.9\% & 25.3\% & 0.0\% & 18.3\% & 18.4\% & 0.0\% \\
env steps to 90\%: TANDEM & 480k & 166k & 353k & 248k & 191k & 284k & 391k & 392k & 404k \\
env steps to 90\%: gate forced closed & 449k & 179k & 487k & 179k & 253k & 284k & 384k & 403k & 404k \\
\bottomrule
\end{tabular}

%% file: tables/table5_training_efficiency.tex
\begin{tabular}{lccccccccc}
\toprule
Method & \makecell{Discovery\\small} & \makecell{Discovery\\larger} & \makecell{Discovery\\control} & \makecell{Sampling\\noise 0.1} & \makecell{Sampling\\noise 0.2} & \makecell{Sampling\\control} & \makecell{Wildfire\\open} & \makecell{Wildfire\\more open} & \makecell{Wildfire\\closed (control)} \\
\midrule
TANDEM & 3.07 & 6.44 & 5.63 & 39.07 & 36.93 & 40.78 & 259.7 & 272.9 & 43.51 \\
MAPPO & 2.35 & 4.16 & 4.95 & 36.56 & 33.56 & 37.80 & 234.1 & 260.5 & 42.87 \\
COMA & 0.81 & 1.53 & 0.81 & 9.33 & 7.10 & 13.29 & 259.2 & 254.6 & 19.29 \\
QMIX & 0.86 & 1.73 & 0.86 & 7.09 & 6.28 & 6.58 & 28.83 & 36.94 & -20.93 \\
MADDPG & 0.79 & 1.81 & 0.82 & 6.72 & 7.18 & 7.75 & 104.2 & 123.0 & -13.86 \\
VDN & 0.87 & 1.64 & 0.76 & 11.45 & 10.24 & 9.85 & 93.24 & 95.61 & -14.39 \\
MAVEN & 0.85 & 1.73 & 0.83 & 7.01 & 6.78 & 7.09 & 41.80 & 28.13 & -21.18 \\
\bottomrule
\end{tabular}